\documentclass[a4paper,UKenglish,cleveref, numberwithinsect, autoref, thm-restate]{lipics-v2021}

\usepackage{cite}

\title{Priority Downward Closures} %

\author{Ashwani Anand}{Max Planck Institute for Software Systems (MPI-SWS), Germany}{ashwani@mpi-sws.org}{https://orcid.org/0000-0002-1825-0097}{}%

\author{Georg Zetzsche}{Max Planck Institute for Software Systems (MPI-SWS), Germany}{georg@mpi-sws.org}{https://orcid.org/0000-0002-6421-4388}{}

\authorrunning{A. Anand and G. Zetzsche} %

\Copyright{Ashwani Anand and Georg Zetzsche} %



\category{} %

\relatedversion{} %

\funding{\flag[3cm]{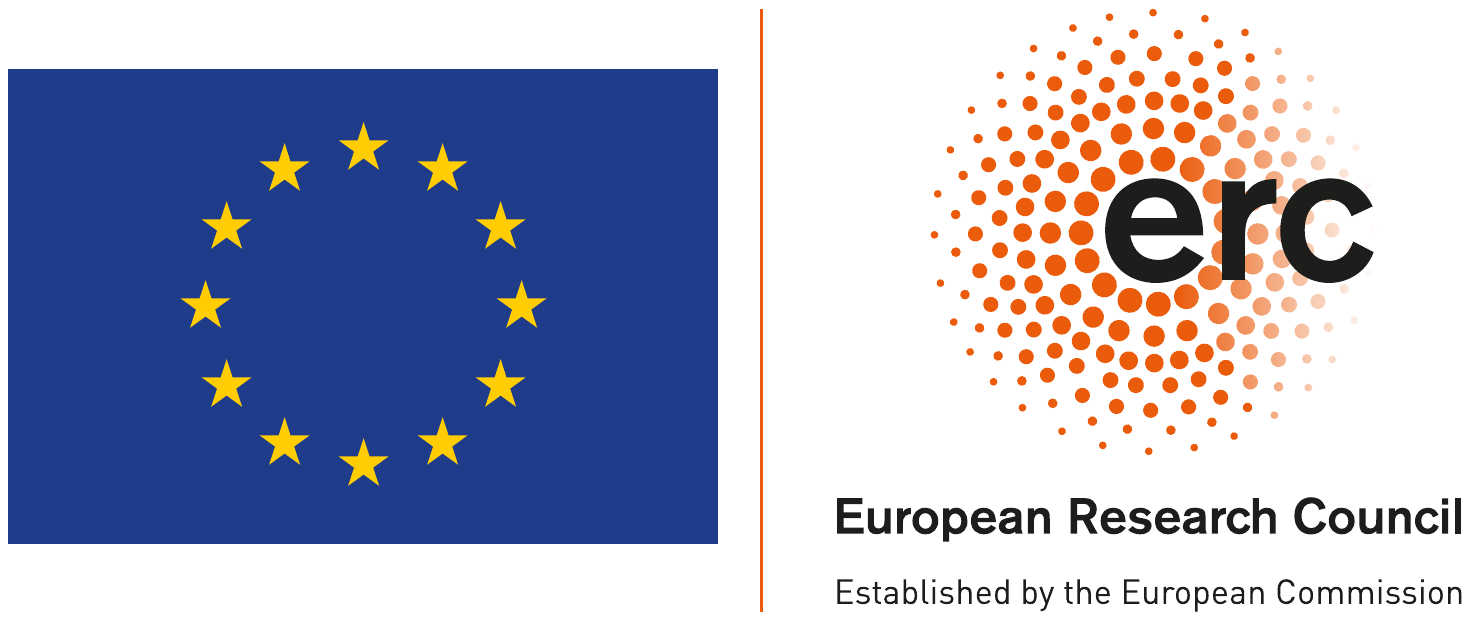}Funded by the European Union (ERC, FINABIS, 101077902). Views and opinions expressed are however those of the author(s) only and do not necessarily reflect those of the European Union or the European Research Council Executive Agency. Neither the European Union nor the granting authority can be held responsible for them.}

\acknowledgements{The authors are grateful to Yousef Shakiba for discussions on the block downward closure of regular languages.}%

\nolinenumbers %

\hideLIPIcs

\EventEditors{Guillermo A. P\'{e}rez and Jean-Fran\c{c}ois Raskin}
\EventNoEds{2}
\EventLongTitle{34th International Conference on Concurrency Theory (CONCUR 2023)}
\EventShortTitle{CONCUR 2023}
\EventAcronym{CONCUR}
\EventYear{2023}
\EventDate{September 18--23, 2023}
\EventLocation{Antwerp, Belgium}
\EventLogo{}
\SeriesVolume{279}
\ArticleNo{39}

\usepackage{amsmath,amssymb,amsfonts}
\usepackage{mathtools}
\usepackage{stmaryrd}
\usepackage{hyperref}

\usepackage{algorithm}
\usepackage{algpseudocode}
\usepackage{changepage}

\usepackage[
colorinlistoftodos, %
textwidth=\marginparwidth, %
textsize=scriptsize, %
]{todonotes}

\usepackage{mdframed}

\usepackage{pgf}
\usepackage{tikz}
\usetikzlibrary{shapes,shapes.geometric,fit,trees,decorations,arrows,automata,shadows,positioning,plotmarks,backgrounds,tikzmark}
\usetikzlibrary{calc,matrix,fit,petri,decorations.markings,decorations.pathmorphing,patterns,intersections,decorations.text}

\usepackage{thm-restate}
\usepackage{xparse}

\usepackage{enumerate}

\usepackage[
colorinlistoftodos, %
textwidth=\marginparwidth, %
textsize=scriptsize, %
]{todonotes}

\usepackage{tikz}
\usetikzlibrary{automata, positioning, arrows, arrows.meta}
\usetikzlibrary{matrix}
\usetikzlibrary{decorations.pathreplacing}
\usetikzlibrary{calc}
\usepackage{subcaption}

\tikzset{->,  %
	>=Latex, %
	node distance=8em, %
	every state/.style={thick}, %
	initial text=$ $, %
	initial/.style={initial by arrow,initial where=above}
}

\newcommand{\sqbrctdown}[5]{
	\draw[#1] (#2.west) -- (#2.south west) -- (#3.south east) node[below, midway,#5]{#4} -- (#3.east)
}

\usepackage{mathtools}
\usepackage{amssymb,stackengine,scalerel}

\usepackage{stmaryrd}
\usepackage{centernot}
\DeclareFontFamily{U}{mathb}{\hyphenchar\font45}
\DeclareFontShape{U}{mathb}{m}{n}{
	<-6> mathb5 <6-7> mathb6 <7-8> mathb7
	<8-9> mathb8 <9-10> mathb9
	<10-12> mathb10 <12-> mathb12
}{}
\DeclareSymbolFont{mathb}{U}{mathb}{m}{n}
\DeclareMathSymbol{\llcurly}{\mathrel}{mathb}{"CE}

\newcommand{\ld}{\lessdot}

\newcommand{\priority}{\ensuremath{\mathcal{P}}}

\renewcommand{\so}{\preccurlyeq}
\NewDocumentCommand\bo{o}{
	\IfNoValueTF{#1}{\preccurlyeq_{\mathsf{B}}}{\preccurlyeq_{\mathsf{B}}^{#1}}
}
\NewDocumentCommand\po{o}{
	\IfNoValueTF{#1}{\preccurlyeq_{\mathsf{P}}}{\preccurlyeq_{\mathsf{P}}^{#1}}
}
\NewDocumentCommand\sd{o}{
	\IfNoValueTF{#1}{\mathord{\downarrow}}{\mathord{\downarrow^{#1}}}
}
\NewDocumentCommand\bd{o}{
	\IfNoValueTF{#1}{\mathord{\downarrow_{\mathsf{B}}}}{\mathord{\downarrow_{\mathsf{B}}^{#1}}}
}
\NewDocumentCommand\pd{o}{
	\IfNoValueTF{#1}{\mathord{\downarrow_{\mathsf{P}}}}{\mathord{\downarrow_{\mathsf{P}}^{#1}}}
}

\newcommand{\Z}{\mathbb{Z}}
\newcommand{\N}{\mathbb{N}}

\newcommand{\s}{\Sigma}
\newcommand{\na}{\mathcal{A}}
\newcommand{\nb}{\mathcal{B}}

\newcommand{\calc}{\mathcal{C}}

\newcommand{\lang}{\mathcal{L}}
\newcommand{\bigO}{\mathcal{O}}

\newcommand{\trans}{\mathcal{T}}

\newcommand{\lna}{\lang(\na)}
\newcommand{\lnb}{\lang(\nb)}

\RequirePackage{xspace}

\newcommand{\wrt}{\hyperlink{wrt}{w.r.t.}\@\xspace}
\newcommand{\wqo}{\hyperlink{wqo}{WQO}\@\xspace}

\newcommand{\FSA}{\hyperlink{FSA}{NFA}\@\xspace}

\newcommand{\colred}{\color{red!60}}
\newcommand{\lowerletters}[1]{\ensuremath{\Sigma_{\leq #1}}}
\newcommand{\equalletters}[1]{\ensuremath{\Sigma_{= #1}}}

\begin{document}

\maketitle

\begin{abstract}
When a system sends messages through a lossy channel, then the language encoding all sequences of messages can be abstracted by its downward closure, i.e. the set of all (not necessarily contiguous) subwords. This is useful because even if the system has infinitely many states, its downward closure is a regular language. However, if the channel has congestion control based on priorities assigned to the messages, then we need a finer abstraction: The downward closure with respect to the priority embedding. As for subword-based downward closures, one can also show that these priority downward closures are always regular.

While computing finite automata for the subword-based downward closure is well understood, nothing is known in the case of priorities. We initiate the study of this problem and provide algorithms to compute priority downward closures for regular languages, one-counter languages, and context-free languages.
\end{abstract}

\section{Introduction}\label{sec:introduction}

When analyzing infinite-state systems, it is often possible to
replace individual components by an overapproximation based on (subword)
downward closures. Here, the \emph{(subword) downward closure} of a language
$L\subseteq\Sigma^*$ is the set of all words that appear as (not necessarily
contiguous) subwords of members of $L$. This overapproximation is usually
possible because the verified properties are not changed when we allow
additional behaviors resulting from subwords. Furthermore, this
overapproximation simplifies the system because a well-known result by Haines
is that for \emph{every language $L\subseteq\Sigma^*$}, its subword downward
closure is regular.

This idea has been successfully applied to many verification tasks, such as the verification of restricted lossy channel systems~\cite{DBLP:conf/icalp/AbdullaBB01}, concurrent programs with dynamic thread spawning and bounded context-switching~\cite{DBLP:journals/corr/abs-1111-1011,DBLP:conf/icalp/0001GMTZ23}, asynchronous programs (safety, termination, liveness~\cite{DBLP:journals/lmcs/MajumdarTZ22}, but also context-free refinement verification~\cite{BaumannGanardiMajumdarThinniyamZetzsche2023b}), the analysis of thread pools~\cite{DBLP:journals/pacmpl/BaumannMTZ22}, and safety of parameterized asynchronous shared-memory systems~\cite{DBLP:conf/concur/TorreMW15}. For these reasons, there has been a substantial amount of interest in algorithms to compute finite automata for subword downward closures of given infinite-state sytems~\cite{vanleeuwen1978,COURCELLE91,DBLP:conf/icalp/HabermehlMW10,DBLP:conf/stacs/Zetzsche15,DBLP:conf/icalp/Zetzsche15,DBLP:conf/lata/BachmeierLS15,DBLP:conf/icalp/Zetzsche16,DBLP:conf/popl/HagueKO16,DBLP:conf/lics/ClementePSW16,DBLP:conf/mfcs/AtigMMS17,DBLP:conf/lics/Zetzsche18,DBLP:journals/corr/abs-1904-10703,Atig2016OCAdown}.

One situation where downward closures are useful is that of systems that send messages through a lossy channel,
meaning that every message can be lost on the way. Then clearly, the downward closure of the set of sequences of messages is exactly the set of sequences observed by the receiver.
This works as long as all messages can be dropped arbitrarily.

\subparagraph{Priorities} However, if the messages are not dropped arbitrarily but as part of congestion control, then taking the set of all subwords would be too coarse an abstraction: Suppose we want to prioritize critical messages that can only be dropped if there are no lower-priority messages in the channel.
For example, RFC~2475 describes an architecture that allows specifying relative priority among the IP packets from a finite set of priorities and allows the network links to drop lower priority packets to accommodate higher priority ones when the congestion in the network reaches a critical point~\cite{Blake1998AnAF}. As another example, in networks with an Asynchronous Transfer Mode layer, cells carry a priority in order to give preferences to audio or video packages over less time-critical packages~\cite{LEBOUDEC1992279}. In these situations, the subword downward closure would introduce behaviors that are not actually possible in the system.

To formally capture the effect of dropping messages by priorities, Haase, Schmitz and Schnoebelen~\cite{HaaseSchmitzSchnoebelen:PowerOfPriorityChannelSystems} introduced \emph{Priority Channel Systems (PCS)}. These feature an ordering on words (i.e.\ channel contents), called the \emph{Prioritised Superseding Order (PSO)}, which allows the messages to have an assigned priority, such that higher priority messages can supersede lower priority ones. This order indeed allows the messages to be treated discriminatively, but the superseding is asymmetric. A message can be superseded only if there is a higher priority letter coming in the channel later. This means, PSO are the ``priority counterpart'' of the subword order for channels with priorities. In particular, in these systems, components can be abstracted by their \emph{priority downward closure}, the downward closure with respect to the PSO. Fortunately, just as for subwords, priority downward closures are also always regular.

This raises the question of whether it is possible to compute finite automata
for the priority downward closure for given infinite-state systems. For
example, consider a recursive program that sends messages into a lossy channel
with congestion control. Then, the set of possible message sequences that can
arrive is exactly the priority downward closure $S\pd$ of the language $S$ of
sent messages. Since $S$ is context-free in this case, we would like to compute
a finite automaton for $S\pd$.  While this problem is well-understood for
subwords, nothing is known for priority downward closures.

\subparagraph{Contribution} We initiate the study of computing priority
downward closures. We show two main results. On the one hand, we study the
setting above---computing priority downward closures of context-free languages.
Here, we show that one can compute a doubly-exponential-sized automaton for its
priority downward closure. On the other hand, we consider a natural restriction
of context-free languages: We show that for one-counter automata, there is a
polynomial-time algorithm to compute the priority downward closure.

\subparagraph{Key technical ingredients} The first step is to consider a
related order on words, which we call \emph{block order}, which also has
priorities assigned to letters, but imposes them more symmetrically. Moreover,
we show that under mild assumptions, computing priority downward closures
reduces to computing block downward closures.

Both our constructions---for one-counter automata and context-free
languages---require new ideas.  For one-counter automata, we modify the
subword-based downward closures construction from~\cite{Atig2016OCAdown} in a non-obvious way to
block downward closures. Crucially, our modification relies on the insight
that, in some word, repeating existing factors will always yield a word that is
larger in the block order.  For context-free languages, we present a novel
inductive approach: We decompose the input language into finitely many
languages with fewer priority levels and apply the construction recursively.

\subparagraph{Outline of the paper} We fix notation in
\cref{sec:preliminaries} and introduce the block order and show its
relationship to the priority order in \cref{sec:block-order}. In
\cref{sec:regular,sec:oca,sec:cfl}, we then present methods for computing block and priority downward closures for regular languages, one-counter languages, and context-free languages, respectively.

\section{Preliminaries}\label{sec:preliminaries}
We will use the convention that $[i,j]$ denotes the set $\{i,i+1,\ldots,j\}$. By $\Sigma $, we represent a finite alphabet. $ \Sigma^*\ (\Sigma^+) $ denotes the set of (non-empty) words over $ \s $. When defining the priority order, we will equip $ \s $ with a set of priorities with total order $ (\priority,\lessdot )$, i.e. there exists a fixed priority mapping from $ \s $ to $ \priority $. The set of priority will be the set of integers $ [0,d] $, with the canonical total order.  By sets $ \equalletters{p} $ ($ p\in\priority $), we denote the set of letters in $ \s $ with priority $ p $. For priority $ p\in\priority $, $ \lowerletters{p}=\equalletters{0}\cup\cdots\cup\equalletters{p} $, i.e. the set of letters smaller than or equal to $ p $. For a word $ w=a_0a_1\cdots a_k $, where $ a_i\in\s $, by $ w[i,j] $, we denote the infix $ a_ia_{i+1}\cdots a_{j-1}a_j $, and by $ w[i] $, we denote $ a_i $.

\subparagraph{Finite automata and regular languages}
A \emph{non-deterministic finite state automaton (\FSA)} is a tuple $ \na= (Q, \Sigma, \delta, q_0, F) $, where
$ Q $ is a finite set of \emph{states},
$ \Sigma $ is its \emph{input alphabet},
$ \delta $ is its set of \emph{edges} i.e.\ a finite subset of $ Q\times \Sigma\cup\{\epsilon\} \times Q $,
$ q_0\in Q $ is its \emph{initial state}, and
$ F\subseteq Q $ is its set of \emph{final states}.
A word is accepted by $ \na $ if it has a run from the initial state ending in a final state. 
The language \emph{recognized} by an NFA $ \na $ is called a regular language, and is denoted by $\lang (\na)$.
The \emph{size of a \FSA}, denoted by $ |\na| $, is the number of states in the \FSA.

\subparagraph{(Well-)quasi-orders} A \emph{quasi-order}, denoted as $ (X, \leq) $, is a set $ X $ with a reflexive and transitive relation $ \leq $ on $ X $. If $ x\leq y $ (or equivalently, $ y\geq x $), we say that $ x $ is smaller than $ y $, or $ y $ is greater than $x$.
If $ \leq $ is also anti-symmetric, then it is called a \emph{partial order}. If every pair of elements in $ X $ is comparable by $ \leq $, then it is called a \emph{total} or \emph{linear} order.
Let $ (X,\leq_1) $ and $ (Y,\leq_2) $ be two quasi orders, and $ h:X\rightarrow Y $ be a function. We call $ h $ a \emph{monomorphism} if it is one-to-one and $ x_1\leq_1 x_2 \iff h(x_1)\leq_2 h(x_2) $. 

A quasi order $ (X, \leq) $ is called a \emph{well-quasi order (\wqo)}, if any infinite sequence of elements $ x_0,x_1,x_2,\ldots $ from $ X $ contains an increasing pair $ x_i\leq x_j $ with $ i<j $. If $X$ is the set of words over some alphabet, then a \wqo $ (X,\leq) $ is called \emph{multiplicative} if $ \forall u,u',v,v'\in X $, $ u\leq u' $ and $ v\leq v' $ imply that $ uv\leq u'v' $.
\subparagraph{Subwords} For $ u,v\in\Sigma^*,$ we say $ u\so v $, which we refer to as \emph{subword order}, if $ u $ is a subword (not necessarily, contiguous) of $ v $, i.e. if 
\begin{eqnarray*}
	u&=&u_1u_2\cdots u_k\\
	\text{and, }v&=& v_0u_1v_1u_2v_2\cdots v_{k-1}u_kv_k
\end{eqnarray*}
where $ u_i\in \s $ and $ v_i\in\s^* $. In simpler words, $ u\so v $ if some letters of $ v $ can be dropped to obtain $ u $. For example, let $ \s=[0,1] $. Then, $ 0\so 00\so 010\not\so 110 $; $ 0 $ and $ 00 $ can be obtained by dropping letters from $ 00 $ and $ 010 $, respectively. But $ 010 $ cannot be obtained from $ 110 $, as the latter does not have sufficiently many $ 0 $s. If $ u\so v $, we say that $ u $ is \textit{subword smaller} than $ v $, or simply that $ u $ is a \textit{subword} of $ v $. And we call a mapping from the positions in $ u $ to positions in $ v $ that witnesses $ u\so v $ as the \emph{witness position mapping}.

Since $ \Sigma $ is a \wqo with the equality order, by Higman's lemma, $ \Sigma^* $ is a \wqo with the subword order. It is in fact a multiplicative WQO: if $ u\so u' $ and $ v\so v' $, then dropping the same letters from $ u'v' $ gives us $ uv $.

\subparagraph{Priority order} We take an alphabet $ \s $ with priorities totally ordered by $ \ld $. We say $ u\po v $, which we refer to as \emph{priority order}, if $ u=\epsilon $ or,
\begin{eqnarray*}
	u&=&u_1u_2\cdots u_k\\
	\text{and, }v&=&v_1u_1v_2u_2\cdots v_k u_k,
\end{eqnarray*}
such that $ \forall i\in [1,k] $, $ u_i\in\s $ and $ v_i\in \lowerletters{u_i}^* $. It is easy to observe that the priority order is multiplicative, and is finer than the subword order, i.e. $ \forall u,v\in\s^*, u\po v\implies u\so v $. %
As shown in \cite[Theorem 3.6]{HaaseSchmitzSchnoebelen:PowerOfPriorityChannelSystems}, the priority order on words over a finite alphabet with priorities is a well-quasi ordering:
\begin{lemma}\label{prioritywqo}
	$ (\Sigma^*, \po) $ is a \wqo.
\end{lemma}

\subparagraph{Downward closure} 
We define the \textit{subword downward closure} and \textit{priority downward closure} for a language $L\subseteq\Sigma^*$ as follows:
\begin{align*}
	L\sd&:=\{ u\in \s^*\, \mid \exists\ v\in L\colon u\ \so\ v  \},  & L\pd &:=\{ u\in \s^*\mid \exists\ v\in L\colon u\ \po\ v  \}.
\end{align*}

The following is the starting point for our investigation: It shows that for every language $L$, there exist finite automata for its downward closures w.r.t.\ $\so$ and $\po$.
\begin{restatable}{lemma}{donwardclosuresregular}\label{subwordregular}
	Every subword downward closed sets and every priority downward closed set is regular.
\end{restatable}
For the subword order, this was shown by Haines~\cite{HAINES196994}. The same idea applies to the priority ordering: A downward closed set is the complement of an upward closed set. Therefore, and since every upward closed set in a well-quasi ordering has finitely many minimal elements, it suffices to show that the set of all words above a single word is a regular language. This, in turn, is shown using a simple automaton construction. In \cref{app:block}, we prove an analogue of this for the block ordering (\cref{generalizedblockregular}).

We stress that \cref{subwordregular} is not effective: It does not guarantee that finite automata for downward closures can be computed for any given language. In fact, there are language classes for which they are not computable, such as reachability sets of \textit{lossy channel systems} and \textit{Church-Rosser languages} \cite{10.1007/10719839_37, GRUBER2007167}.  Therefore, our focus will be on the question of how to effectively compute automata for priority downward closures.

\section{The Block Order}\label{sec:block-order}

We first define the block order formally and then give the intuition behind the definition.
Let $ \Sigma $ be a finite alphabet, and $\priority= [0,d] $ be a set of priorities with a total order $ \ld $. 
Then for $ u,v\in \Sigma^* $, where maximum priority occurring among $ u $ and $ v $ is $ p $, we say $ u\bo v $, if
\begin{enumerate}[i.]
	\item if $ u,v\in \equalletters{p}^* $, and $ u\so v $, or
	\item if 
	\begin{eqnarray*}
		u&=& u_0x_0u_1x_1\cdots x_{n-1}u_{n}\\
		\text{and, }v&=& v_0y_0v_1y_1\cdots y_{m-1}v_{m}
	\end{eqnarray*}
	where $  x_0,\ldots x_{n-1},y_0,\ldots, y_{m-1} \in \equalletters{p}$, and for all $i\in[0,n]$, we have $ u_i,v_i\in \lowerletters{p-1}^*$ (the $u_i$ and $v_i$ are called \emph{sub-$ p $} blocks), and there exists a strictly monotonically increasing map $ \phi:[0,n]\rightarrow [0,m] $, which we call the \emph{witness block map}, such that 
	\begin{enumerate}
		\item $ u_i\bo v_{\phi(i)} $, $ \forall i $,\label{def:generalized is recursive}
		\item $ \phi(0)= 0 $, \label{def:first-to-first}
		\item $ \phi(n) =m $, and\label{def:last-to-last}
		\item $ x_i\so v_{\phi(i)}y_{\phi(i)}v_{\phi(i)+1}\cdots v_{\phi(i+1)} $, $ \forall i \in [0,n-1] $.\label{def:generalized is subword}
	\end{enumerate}
\end{enumerate}

Intuitively, we say that $ u $ is \emph{block smaller} than $ v $, if either
\begin{itemize}
	\item both words have letters of same priority, and $ u $ is a subword of $ v $, or,
	\item the largest priority occurring in both words is $ p $. Then we split both words along the priority $ p $ letters, to obtain sequences of sub-$ p $ blocks of words, which have words of strictly less priority. Then by item \ref{def:generalized is recursive}, we embed the sub-$ p $ blocks of $ u $ to those of $ v $, such that they are recursively block smaller. Then with items \ref{def:first-to-first} and \ref{def:last-to-last}, we ensure that the first (and last) sub-$ p $ block of $ u $ is embedded in the first (resp., last) sub-$ p $ block of $ v $. We will see later that this constraint allows the order to be multiplicative. Finally, by item \ref{def:generalized is subword}, we ensure that the letters of priority $ p $ in $ u $ are preserved in $ v $, i.e. every $ x_i $ indeed occurs between the embeddings of the sub-$ p $ block $ u_{i} $ and $ u_{i+1} $.
\end{itemize}

\begin{example}\label{ex:block-order}
	Consider the alphabet $ \s= \{ 0^a,0^b,1^a,1^b,2^a,2^b   \} $ with priority set $ \priority=[0,2] $ and $ \equalletters{i}=\{i^a, i^b\} $. 
	In the following examples, the {\colred color} helps to identify the largest priority occurring in the words.
	First, notice that $ \epsilon\bo 0^a\bo 0^a0^b $, and hence 
	\begin{align*} &{\colred 1^b}0^a\bo0^a{\colred 1^b}0^a0^a1^a0^a0^b, &&\text{but} && {\colred 1^b}0^a\not\bo0^a{\colred 1^b}0^a0^a{\colred 1^a}0^b0^b. \end{align*}
	This is because $ 0^a\not\bo 0^b0^b $, i.e. the last sub-$ 1 $ block of the former word cannot be mapped to the last sub-$ 1 $ block of the latter word. As another example, we have 
	\begin{align*} &{\colred 2^a}1^b0^a\bo 0^a{\colred 2^a}0^a1^b0^a0^a1^a0^a0^b, &&\text{but} &  {\colred 2^a}1^b0^a\not\bo 0^a{\colred 2^b}0^a1^b0^a0^a1^a0^a0^b. \end{align*}
	This is because $ 2^a $ does not exist in the latter word, violating item \ref{def:generalized is subword}. 
	Finally, notice that
	\begin{equation} 1^a1^b\not\bo 1^a {\colred 2^a} 1^b, \label{no-insertions}\end{equation}
	because the sub-$2$ block $1^a1^b$ would have to be mapped to a single sub-$2$ block in the right-hand word; but none of them can accomodate $1^a1^b$.
\end{example}

Note that by items \ref{def:generalized is subword} and \ref{def:generalized is recursive}, we have that $ u\bo v\implies u\so v $, for all $ u,v\in \s^* $. Then there exists a position mapping $ \rho $ from $ [0,|u|] $ to $ [0,|v|] $ such that $ u[i]=v[\rho(i)] $, for all $ i $. We say that a position mapping \emph{respects block order} if for all $ i $, $ v[\rho(i),\rho(i+1)] $ contains letters of priorities smaller than $ u[i] $ and $ u[i+1] $. It is easy to observe that if $ u\bo v $, then there exists a position mapping from $ u $ to $ v $ respecting the block order. 
The following is a straightforward repeated application of Higman's Lemma~\cite{10.1112/plms/s3-2.1.326} (see \cref{app:block}).
\begin{restatable}{theorem}{thmgeneralizedblockwqo}\label{thm:generalizedblockwqo}
	$ (\Sigma^*, \bo) $ is a \wqo.
\end{restatable}

In fact, the block order is multiplicative, i.e. for all $ u,v,u',v'\in\s^* $ such that $ u\bo u' $ and $ v\bo v' $, it holds that $ uv\bo u'v' $.
\begin{lemma}\label{generalizedblockmultiplicative}
	$ (\Sigma^*,\bo) $ is a multiplicative \wqo.
\end{lemma}
\begin{proof}
	For singleton $ \priority $, the result trivially holds because it coincides with the subword order. Let $ (\lowerletters{p-1}^*,\bo) $ be multiplicative. Now we show that $ (\lowerletters{p}^*,\bo) $ is multiplicative. To this end, let $ u\bo u' $, $ v\bo v' $, and $ \phi, \psi $ be the witnessing block maps respectively. We assume
	\begin{eqnarray*}
		u&=& u_0x_0u_1x_1u_2x_2\cdots x_{k-1}u_k\\
		v&=& v_0y_0v_1y_1v_2y_2\cdots y_{l-1}v_l\\
		u'&=& u'_0x'_0u'_1x'_1u'_2x'_2\cdots x'_{k-1}u'_{k'}\\
		v'&=& v'_0y'_0v'_1y'_1v'_2y'_2\cdots y'_{l-1}v'_{l'}
	\end{eqnarray*}
	where $ x_i,y_i,x_i',y_i'\in \equalletters{p} $. 
	Consider the function $ \delta\colon [0,k+l-1]\rightarrow [0,k'+l'-1]$ with
	\[ i\mapsto \begin{cases}
		\phi(i), \text{ if } 1\leq i\leq k\\
		\psi(i-k+1), \text{ if } k< i\leq k+l-1
	\end{cases} \]
	Since the $ k^{th} $ sub-$ p $ block of $ u $ and the $ 1^{st} $ sub-$ p $ block of $ v $ combines in $ uv $ to form one sub-$ p $ block, we have $ k+l-1 $ sub-$ p $ blocks. Similarly, $ u'v' $ has $ k'+l'-1 $ sub-$ p $ blocks. And hence $ u_kv_1\bo u'_{k'}v'_1 $, by induction hypothesis. The recursive embedding is obvious for other sub-$ p $ blocks. We also have that $ \delta(0)=0 $ and $ \delta(k+l-1)=k'+l'-1 $. By monotonicity of $ \phi $ and $ \psi $, $ \delta $ is also strictly monotonically increasing. Hence, $ \delta $ witnesses $ uv\bo u'v' $.
\end{proof}

\subparagraph{Pumping} 
In the subword ordering, an often applied property is that for any words
$u,v,w$, we have $uw\so uvw$, i.e.\ inserting any word leads to a superword.
This is not true for the block ordering, as we saw in \cref{ex:block-order},
\eqref{no-insertions}. However, one of our key observations about the block
order is the following property: If the word we insert is just a repetition of
an existing factor, then this yields a larger word in the block ordering. This
will be crucial for our downward closure construction for one-counter automata
in \cref{sec:oca}. 
\begin{restatable}[Pumping Lemma]{lemma}{arbitrarypumpinglemma}\label{generalizedblockrepeat}
	For any $u,v,w\in\Sigma^*$, we have $uvw\bo uvvw$.
\end{restatable}
Before we prove \cref{generalizedblockrepeat}, let us note that by applying
\cref{generalizedblockrepeat} multiple times, this implies that we can also
repeat multiple factors. For instance, if $w=w_1w_2w_3w_4w_5$, then $w\bo
w_1w_2^2w_3w_4^3w_5$. \Cref{fig:sborepeat} shows an example on how to choose
the witness block map.
\begin{proof}
	We proceed by induction on the number of priorities. If there is just a
	single priority (i.e. $\priority=\{0\}$), then $\bo$ coincides with $\so$
	and the statement is trivial. Let us assume the \lcnamecref{generalizedblockrepeat} is established for words with up to $n$ priorities. We distinguish two cases.
	
	\begin{itemize}
		\item Suppose $v$ contains only letters of priorities $[0,n]$. Then repeating
	$v$ means repeating a factor inside a sub-$(n+1)$ block, which is a
	word with priorities in $[0,n]$. Hence, the statement follows by
	induction: Formally, this means we can use the embedding mapping that
	sends block $i$ of $uvw$ to block $i$ of $uvvw$.

	\item Suppose $v$ contains a letter of priority $n+1$. write
		$v=v_0x_1v_1\cdots x_mv_m$, where $x_1,\ldots,x_m$ are the letters of
	priority $n+1$ in $v$ and $v_0,\ldots,v_m$ are the sub-$(n+1)$ blocks
	of $v$. Then:
	\begin{align*} uvw&=uv_0x_1\cdots v_{m-1}x_mv_mw, &
		uvvw&=uv_0x_1\cdots v_{m-1}x_m\underbrace{v_mv_0x_1\cdots v_{m-1}x_m}_{\text{skipped}}v_mw. \end{align*}
	The idea is simple: Our witness block map just skips the $m$ sub-$(n+1)$ blocks
	inside of $v_mv_0x_1\cdots v_{m-1}x_m$. Thus, the sub-$(n+1)$ blocks in $uv_0x_1\cdots v_{m-1}x_m$ are mapped to the same blocks in $uv_0x_1\cdots v_{m-1}x_m$, and the sub-$(n+1)$ blocks in $v_mw$ are mapped to the same blocks in $v_mw$.
	This is clearly a valid witness block map, since the first (resp.\
	last) sub-$(n+1)$ block is mapped to the first (resp.\ last), and each
	sub-$(n+1)$ block is mapped to an identical sub-$(n+1)$ block. \qedhere
	\end{itemize}

\end{proof}
	\begin{figure}[t]
		\centering
		\begin{tikzpicture}
			\matrix (n) [matrix of math nodes, row sep=1.7em,
			column sep=0.5,
			row 1 column 4/.style={red},
			row 1 column 5/.style={red},
			row 1 column 6/.style={red},
			row 1 column 7/.style={red},
			row 1 column 10/.style={red},
			row 1 column 11/.style={red},
			row 1 column 12/.style={red},
			row 1 column 13/.style={red},
			row 1 column 14/.style={red},
			row 1 column 15/.style={red},
			row 1 column 16/.style={red},
			row 1 column 17/.style={red},
			row 1 column 18/.style={red},
			row 2 column 8/.style={red},
			row 2 column 9/.style={red},
			row 2 column 12/.style={red},
			row 2 column 13/.style={red},
			row 2 column 14/.style={red}
			]{
				w'=&1&2&0&1&0&1&2&1&1&2&1&1&2&1&1&2&1&0\\
				w=&&&&&1&2&0&1&2&1&1&2&1&0&&&&\\
			};
			
			\sqbrctdown{-}{n-1-2}{n-1-2}{}{};
			\sqbrctdown{-}{n-1-4}{n-1-7}{}{};
			\sqbrctdown{-}{n-1-9}{n-1-10}{}{};
			\sqbrctdown{-}{n-1-12}{n-1-13}{}{};
			\sqbrctdown{-}{n-1-15}{n-1-16}{}{};
			\sqbrctdown{-}{n-1-18}{n-1-19}{}{};
			\sqbrctdown{-}{n-2-6}{n-2-6}{}{};
			\sqbrctdown{-}{n-2-8}{n-2-9}{}{};
			\sqbrctdown{-}{n-2-11}{n-2-12}{}{};
			\sqbrctdown{-}{n-2-14}{n-2-15}{}{};

			\draw [-Latex, gray!60] (n-2-6.north)-- (n-1-2.south);
			\draw [-Latex, gray!60] (n-2-8.north east)-- (n-1-5.south east);
			\draw [-Latex, gray!60] (n-2-11.north east)-- (n-1-9.south east);
			\draw [-Latex, gray!60] (n-2-14.north east)-- (n-1-19.south west);
		\end{tikzpicture}

		\caption[Generalized block order repetition mapping]{Here $ \Sigma=[0,2] $, $ \priority=[0,2] $, and $ A_i=\{i\} $, $ w=12(01)21(121)0 $ and $ w'=12(01)^221(121)^30 $. The repeated segments are marked in {\color{red} red}, and the arrows denote the witness block map.}
		\label{fig:sborepeat}
	\end{figure}
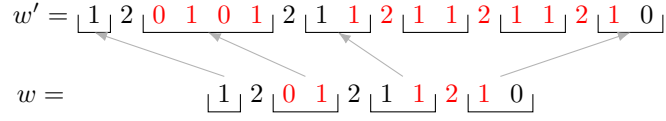

\subparagraph{Regular downward closures}
As for $\so$ and $\po$, we define $L\bd=\{u\in\Sigma^* \mid \exists v\in L\colon u\bo v\}$ for any $L\subseteq \Sigma^*$.

\begin{restatable}{lemma}{lemgeneralizedblockregular}\label{generalizedblockregular}
	For every $L\subseteq\Sigma^*$, $L\bd$ is a regular language.
\end{restatable}
For the proof of \cref{generalizedblockregular}, one can argue as mentioned above: The complement $\Sigma^*\setminus (L\bd)$ of $L\bd$ is upward closed. And since $\bo$ is a WQO, $\Sigma^*\setminus (L\bd)$ has finitely many minimal elements. It thus remains to show that for each word $w\in\Sigma^*$, the set of words $\bo$-larger than $w$ is regular, which is a simple exercise. Details can be found in \cref{app:block}.

\subparagraph{Block order vs.\ priority order}
We will later see (\cref{generalizedToPriority}) that under mild conditions, computing priority downward closures reduces to computing block downward closures. The following \lcnamecref{generalizedBlockFiner} is the main technical ingredient in this:
It shows that the block order refines the priority order on words that end in the same letter, assuming the alphabet has a certain shape. A priority alphabet $(\Sigma,\priority)$ with $\priority=[1,d]$ is called \emph{flat} if $|\equalletters{i}|=1$ for each $i\in [1,d]$.
\begin{lemma}\label{generalizedBlockFiner}
	If $ \s $ is flat and $ u,v\in \Sigma^*a$ for some $a\in\s $, then $ u\bo v$ implies $u\po v$. 
\end{lemma}
\begin{proof}
	Since $ u\bo v $, there exists a witness position mapping $ \rho $ that maps the positions of the letters in $ u $ to that of $ v $, such that it respects the block order, and it maps the last position of $ u $ to the last of $ v $. 
	
	Let $ u=u_0u_1\cdots u_k $. We say that a position mapping violates the priority order at position $ i $ (for $ i\in [0,k-1] $), if $ v[\rho(i)+1, \rho(i+1)] $ has a letter of priority higher than that of $ u[i+1] $. Note that if $ \rho $ does not violate the priority order at any position, then $ u\po v $. %
	
	Let $ i $ be the largest position at which $ \rho $ violates the priority order, i.e.  $ v[\rho(i)+1, \rho(i+1)] $ has a letter of priority higher than that of $ u[i+1] $. We show that if $ \rho $ respects the block order till position $ i $, there exists another witness position mapping $ \rho' $ that respects the block order till position $ i-1 $, and has one few position of violation (i.e. no violation at position $ i $).

	We first observe that $ u[i]>u[i+1] $, which holds since $ \rho $ respects the block order till position $ i $, implying that $ v[\rho(i)+1, \rho(i+1)] $ does not have a letter of priority higher than $ min\{u[i], u[i+1]\} $, and if $ u[i]\leq u[i+1] $, $ \rho $ does not violate the priority order at $ i $.%

	Then observe that $  v[\rho(i)+1, \rho(i+1)]  $ does not have a letter with priority $ p $, where $ u[i]>p> u[i+1] $, otherwise the sub-$ u[i] $ block of $ u $ immediately after $ u[i] $, can not be embedded to that of $ v $ immediately after $ v[\rho(i)] $, since it would have to be split along $ p $, and the first sub-$ p $ block in $ v $ will not be mapped to any in $ u $. Then $  v[\rho(i)+1, \rho(i+1)]  $ has letter of priority $ u[i] $ (for a violation at $ i $). Then consider the mapping $ \rho' $ that maps $ i $ to the last $ u[i] $ letter in $  v[\rho(i)+1, \rho(i+1)]  $ (say at $ v[j] $ for some $j$, $ \rho(i)+1\leq j\leq\rho(i+1) $). 
	
	This mapping respects the block order till position $ i-1 $, trivially, as we do not change the mapping before $ i $. We show that there is no priority order violation at position $ i $. This holds because the only larger priority letter occurring in $ v[\rho(i)+1, \rho(i+1)] $ was $ u[i] $, and due to the definition of $ \rho' $, $ v[\rho'(i)+1, \rho'(i+1)]  $ has no letter of priority higher than $ u[i+1] $. Since we do not change the mapping after position $ i $, $ \rho' $ does not introduce a violation at any position after $ i $. Hence we have a new position mapping that has one few position of priority order violation.
\end{proof}

\begin{remark}
	We want to stress that the flatness assumption in \cref{generalizedBlockFiner} is crucial: Consider the alphabet $ \s $ from the \cref{ex:block-order}. Then $ 1^a0^a\bo 1^a1^b0^a $, but $ 1^a0^a\not\po 1^a1^b0^a $. Here only one position mapping exists, and it is not possible to remap $ 1^a $ to $ 1^b $ since they are two distinct letters of same priority. Hence, we need to assume that each priority greater than zero has at most one letter.
\end{remark}

\section{Regular Languages}\label{sec:regular}

In this section, we show how to construct an NFA for the block downward closure of a regular language. To this end, we show that both orders are rational transductions.

\subparagraph{Rational transductions}
A \emph{finite state transducer} is a tuple $ \na=(Q,X,Y,E,q_0,F) $, where %
$ Q $ is a finite set of states,
$ X $ and $ Y $ are \emph{input} and \emph{output alphabets}, respectively,
$ E $ is the set of \emph{edges} i.e. finite subset of $ Q\times X^*\times Y^*\times Q $,
$ q_0\in Q $ is the \emph{initial state}, and
$ F\subseteq Q $ is the set of \emph{final states}.
A \emph{configuration} of $ \na $ is a triple $ (q,u,v)\in Q\times X^*\times Y^* $. We write $ (q,u,v)\rightarrow_{\na} (q',u',v')  $, if there is an edge $ (q,x,y,q') $ with $ u'=ux $ and $ v'=vy $. If there is an edge $ (q,x,y,q') $, we sometimes denote this fact by $ q\xrightarrow{(x,y)}_{\na} q' $, and say ``read $ x $ at $ q $, output $ y $, and goto $ q' $''. The \emph{size of a transducer}, denoted by $ |\na| $, is the number of its states.

A \emph{transduction} is a subset of $ X^*\times Y^* $ for some finite alphabets $ X,Y $. The \emph{transduction defined by $ \na $} is $  \trans(\na) = \{ (u,v)\in X^*\times Y^* \ |\ (q_0,\epsilon,\epsilon)\rightarrow_{\na}^* (f,u,v)\text{ for some } f\in F     \}.  $ 
A transduction is called \emph{rational} if it is defined by some finite-state transducer. Sometimes we abuse the notation and output a regular language $ R\subseteq Y^* $ on an edge, instead of a letter. It should be noted that this abuse is equivalent to original definition of finite state transducers.

We say that a language class $ \calc $ is \emph{closed under rational transductions} if for each language $ L \in \calc $, and each rational transduction $ R\subseteq X^*\times Y^* $, \emph{the language obtained by applying the transduction} $ R $ to $ L $, $  RL \stackrel{def}{=} \{ v\in Y^*\ |\ (u,v)\in R \text{ for some } u\in L  \}  $  also belongs to $ \calc $. We call such language classes \emph{full trio}%
. Regular languages, context-free languages, recursively enumerable languages are some examples of full trios~\cite{berstel1979transductions}.%

\subparagraph{Transducers for orders}
It is well-known that the subword order is a rational transduction, i.e. the relation $ T = \{ (u,v)\in X^*\times X^*\ |\ v\so u   \} $ is defined by a finite-state transducer. For example, it can be defined by a one-state transducer that can non-deterministically decide to output or drop each letter. Note that on applying the transduction to any language, it gives the subword downward closure of the language. This means, for every $L\subseteq X^*$, we have $TL=L\sd$.
We will now describe analogous transducers for the priority and block order. %
\begin{restatable}{theorem}{sizeprioritytrans}\label{sizeprioritytrans}
	Given a priority alphabet with priorities $[0,k]$, one can construct in polynomial time a transducer for $\bo$ and a transducer for $\po$, each of size $\bigO(k)$.
\end{restatable}
\begin{proof}
	The transducers for the block and priority order are similar. Intuitively, both remember the maximum of the priorities dropped or to be dropped, and keep or drop the coming letters accordingly. We show the transducer for the priority order here since it is applied in \cref{generalizedToPriority}. The transducer for the block order is detailed in \cref{app:regular}. 
	
	Let $ \s $ be a finite alphabet, with priorities $ \priority= [0,k] $.
	Consider the transducer that has one state for every priority, a non-final sink state, and a distinguished final state. If the transducer is in the state for priority $r$ and reads a letter $a$ of priority $s$, then
	\begin{itemize}
		\item if $s<r$, then it outputs nothing and stays in state $r$, 
		\item if $s\ge r$, then it can output nothing, and go to state $s$,
		\item if $s\ge r$, it can also output $a$, and go to state $0$, or the accepting state non-deterministically,
		\item for any other scenario, goes to the sink state.
	\end{itemize}
	The priority $ 0 $ state is the initial state. Intuitively, the transducer remembers the largest priority letter that has been dropped, and keeps only a letter of higher priority later. To be accepting, it has to read the last letter to go to the accepting final state. 
\end{proof}

The following theorem states that the class of regular languages form a full trio.
\begin{theorem}[\protect{\cite[Corollary 3.5.5]{shallit_2008}}]\label{regularFullTrio}
	Given an \FSA $\na$ and a transducer $\nb$, we can construct in polynomial time an \FSA of size $|\na|\cdot|\nb|$ for $\trans(\nb)(\lang(\na))$.
\end{theorem} 
\Cref{regularFullTrio,sizeprioritytrans} give us a polynomial size \FSA recognizing the priority and block downward closure of a regular language, which is computable in polynomial time as well.

\begin{theorem}\label{prioritydownwardregular}
	Priority and block downward closures for regular languages are effectively computable in time polynomial in the number of states in the \FSA recognizing the language.
\end{theorem}

\cref{prioritydownwardregular,generalizedBlockFiner} now allow us to reduce the
priority downward closure computability to computability for block order. 
\begin{theorem}\label{generalizedToPriority}
	If $ \calc $ is a full trio and we can effectively compute block downward closures for $\calc$, then we can effectively compute priority downward closures.
\end{theorem}
\begin{proof}
	The key idea is to reduce priority downward closure computation to the setting where (i)~all words end in the same letter and (ii)~the alphabet is flat. Since by \cref{generalizedBlockFiner}, on those languages, the block order is finer than the priority order, computing the block order will essentially be sufficient.
	
	Let us first establish (i). Let $ L\in\calc $. Then for each $a\in\s$, the language $L_a=L\cap
	\s^*a$ belongs to $\calc$. Since $L=\bigcup_{a\in \s} L_a\cup E$ and thus
	$L\pd=\bigcup_{a\in\s} L_a\pd\cup E$, it suffices to compute priority downward
	closures for each $L_a$, where $ E = \{\epsilon\} $ if $ \epsilon\in L $, else $ \emptyset $. This means, it suffices to compute priority downward
	closures for languages where all words end in the same letter.
	
	To achieve (ii), we make the alphabet flat. We say that $ (\s,\priority') $ is the \emph{flattening} of $( \s,\priority= [0,d]) $, if $ \priority' $ is obtained by choosing a total order to $\Sigma$ such that if $a$ has smaller priority than $b$ in $(\s,\priority)$, then $a$ has smaller priority than $b$ in $(\s,\priority')$. (In other words, we pick an arbitrary linearization of the quasi-order on $\Sigma$ that expresses ``has smaller priority than''). Then, we assign priorities based on this total ordering. Let $\bo[\mathsf{flat}]$ and $\po[\mathsf{flat}]$ denote the block order and priority order, resp., based on the flat priority assignment. It is a simple observation that for $u,v\in\Sigma^*$, we have that $u\po[\mathsf{flat}] v$ implies $u\po v$. 
	
	Now observe that for $u,v\in L_a$, \cref{generalizedBlockFiner} tells us that $u\bo[\mathsf{flat}] v$ implies $u\po[\mathsf{flat}] v$ and therefore also $u\po v$. This implies that $(L_a\bd[\mathsf{flat}])\pd=L_a\pd$. By assumption, we can compute a finite automaton $\na$ with $\lang(\na)=L_a\bd[\mathsf{flat}]$. Since then $\lang(\na)\pd=(L_a\bd[\mathsf{flat}])\pd=L_a\pd$,  we can compute $L_a\pd$ by applying \cref{prioritydownwardregular} to $\na$ to compute $\lang(\na)\pd=L_a\pd$.
\end{proof}

\section{One-counter Languages}\label{sec:oca}

In this section, we show that for the class of languages accepted by one-counter automata, which form a full-trio \cite[Theorem 4.4]{berstel1979transductions}, the block and priority downward closures can be computed in polynomial time. We prove the following theorem.
\begin{theorem}\label{blockdownwardOCA}
	Given an OCA $\na$, $ \lang(\na)\bd $ and $\lang(\na)\pd$ are computable in polynomial time.
\end{theorem}
Here, the difficulty is that existing downward closure constructions exploit
that inserting any letters in a word yields a super-word. However, for
the block order, this might not be true: Introducing high-priority letters
might split a block unintentionally. However, we observe that the subword
closure construction from~\cite{Atig2016OCAdown} can be modified so that when
constructing larger runs (to show that our NFA only accepts words in the
downward closure), we only repeat existing factors. 
\Cref{generalizedblockrepeat} then yields that the resulting word is
block-larger.

According to \cref{generalizedToPriority}, it suffices to show that block
downward closures are computable in polynomial time (an inspection of the proof
of \cref{generalizedToPriority} shows that computing the priority downward
closure only incurs a polynomial overhead).

\subparagraph{One-counter automata.} One-counter automata are finite state automata with a counter that can be incremented, decremented, or tested for zero. Formally,
a \emph{one-counter automaton (OCA)} $ \na $ is a $ 5 $-tuple $ (Q,\Sigma,\delta,q_0,F) $ where $ Q $ is a finite set of states, $ q_0\in Q $ is an initial state, $ F\subseteq Q $ is a set of final states, $ \Sigma $ is a finite alphabet and $ \delta\subseteq Q\times(\Sigma\cup\{\epsilon \} )\times \{-1,0,+1,z \}\times Q $ is a set of transitions. Transitions $ (p_1,a,s,p_2)\in\delta $ are classified as \emph{incrementing} $ (s=+1) $, \emph{decrementing} $ (s=-1) $, \emph{internal} $ (s=0) $, or \emph{test for zero}$ (s=z) $. 

A \emph{configuration} of an $ OCA $ is a pair that consists of a state and a (non-negative) counter value, i.e., $ (q,n)\in Q\times \N $. 
A sequence $ \pi= (p_0,c_0),t_1,(p_1,c_1),t_2,\cdots, t_m,(p_m,c_m) $ where $ (p_i,c_i)\in Q\times\Z $, $ t_i\in\delta $ and $ (p_{i-1},c_{i-1})\xrightarrow{t_i} (p_i,c_i) $ is called:
\begin{itemize}
	\item a \emph{quasi-run}, denoted $ \pi=(p_0,c_0)\xRightarrow{w}_\na(p_m,c_m)$, if none of $ t_i $ is a test for zero;
	\item a \emph{run}, denoted $ \pi=(p_0,c_0)\xrightarrow{w}_\na(p_m,c_m) $, if all $ (p_i,c_i)\in Q\times\N $.
\end{itemize}
For any quasi-run $ \pi $ as above, the sequence of transitions $ t_1,\cdots,t_m $ is called a \emph{walk} from the state $ p_0 $ to the state $ p_m $.
A run $ (p_0,c_0)\xrightarrow{w}(p_m,c_m) $ is called \emph{accepting} in $ \na $ if $ (p_0,c_0)=(q_0,0) $ where $ q_0 $ is the initial state of $ \na $ and $ p_m $ is a final state of $ \na $, i.e. $ p_m\in F $. In such a case, the word $ w $ is \emph{accepted} by $ \na $.

\subparagraph{Simple one-counter automata} As we will show later, computing
block downward closures of OCA easily reduces to the case of simple OCA. A
\emph{simple OCA (SOCA)} is defined analogously to OCA, with the differences that
(i)~there are no zero tests, (ii)~there is only one final state, (iii)~for acceptance, the final counter value must be zero.

We first show that the block downward closures can be effectively computed for the simple one-counter automata languages.

\begin{restatable}{proposition}{blockdownwardsimpleOCA}\label{blockdownwardsimpleOCA}
	Given a simple OCA $ \na $, we can compute $ \lang(\na)\bd $ in polynomial time.
\end{restatable}
We present a rough sketch of the construction, full details can be found in \cref{app:oca}. The starting point of the construction is the one for subwords
in~\cite{Atig2016OCAdown}, but the latter needs to be modified in a non-obvious
way using \cref{generalizedblockrepeat}.

Let $ \na =(Q,\Sigma,\delta,q_0,q_f) $ be a simple OCA, with $ |Q|=K $. We construct an \FSA\ $\nb$ that can simulate $\na$ in three different modes. In the first mode, it simulates $ \na $ until the counter value reaches $ K $, and when the value reaches $ K+1 $, it switches to the second mode. The second mode simulates $ \na $ while the counter value stays below $ K^2+K+1 $. 
Moreover, and this is where our construction differs from \cite{Atig2016OCAdown}: if $\nb$ is in the second mode simulating $\na$ in some state $q$, then $\nb$ can spontaneously execute a loop from $q$ to $q$ of $\na$ while ignoring its counter updates. When the counter value in the second mode drops to $ K $ again, $\nb$ non-deterministically switches to the third mode to simulate $ \na $ while the counter value stays below $ K $. Thus, $\nb$ only needs to track counter values in $[0,K^2+K+1]$, meaning they can be stored in its state. We claim that then $\lang(\na)\subseteq\lang(\nb)\subseteq\lang(\na)\bd$.
\begin{restatable}{lemma}{ocacontainsoriginaloca}\label{lemma:intermediate automata contains original automata}
	$ \lna\subseteq \lnb $.
\end{restatable}
If a word in $\lang(\na)$ has a run with counters bounded by $K^2+K+1$, then it trivially belongs to $\lang(\nb)$.
If the counters go beyond $K^2+K+1$, then with the classical ``unpumping'' argument, one can extract two loops, one increasing the counter, one decreasing it. These loops can then be simulated by the spontaneous loops in the second mode of $\nb$.

The more interesting inclusion is the following:
\begin{restatable}{lemma}{ocadownwardclosurecontainsoca}\label{lem:ocadownwardclosurecontainsoca}
	$ \lnb\subseteq \lna\bd $.
\end{restatable}
We have to show that each spontaneous loop in $\nb$ can be justified by padding the run with further loop executions so as to obtain a run of $\na$. This is possible because to execute such a spontaneous loop, we must have gone beyond $K$ and later go to zero again. Thus, there exists a ``pumping up'' loop adding, say $k\ge 0$ to the counter, and a ``pumping down'' loop, subtracting, say $\ell\ge 0$ from the counter. We can therefore repeat all spontaneous loops so often that their effect --- when seen as transitions in $ \na $ --- is a (positive or negative) multiple $M$ of $k\cdot\ell$. Then, we execute the $k$- and the $\ell$-loop so often so as to get the counter values so high that (i)~our repeated spontaneous loops never cross zero and (ii)~the effect difference of the new loops is exactly $M$. Since in our construction (in contrast to \cite{Atig2016OCAdown}), the padding only \emph{repeated words that already exist} in the run of $\nb$, \cref{generalizedblockrepeat} implies that the word of $\nb$ embeds via the block order.

\subparagraph{General OCA} Let us now show how to construct the block downward
closure of general OCAs. Suppose we are given an OCA $\na$. For any two states
$p,q$, consider the simple OCA $\na_{p,q}$ obtained from $\na$ by removing all
zero tests, making $p$ initial, and $q$ final. Then $\lang(\na)$ is the set of
words read from $(p,0)$ to $(q,0)$ without using zero tests. We now compute for
each $p,q$ a finite automaton $\nb_{p,q}$ for the block downward closure of
$\na_{p,q}$. Clearly, we may assume that $\nb_{p,q}$ has exactly one initial
state and one final state. Finally, we obtain the finite automaton $\nb$  from
$\na$ as follows: We remove all transitions \emph{except} the zero tests. Each
zero test from $p$ to $q$ is replaced with an edge
$p\xrightarrow{\varepsilon}q$. Moreover, for any states $p$ and $q$ coming from
$\na$, we glue in the automaton $\nb_{p,q}$ (by connecting $p$ with
$\nb_{p,q}$'s initial state and connecting $\nb_{p,q}$'s final state with $q$).
Then, since the block order is multiplicative, we have that $L(\nb)$ accepts
exactly the block downward closure of $\na$.

Futhermore, note that since our construction for simple OCA is polynomial, the
general case is as well: The latter employs the former to $|Q|^2$ simple OCAs.

\section{Context-free Languages}\label{sec:cfl}

\newcommand{\cG}{\mathcal{G}}
\newcommand{\cH}{\mathcal{H}}
\newcommand{\cE}{\mathcal{E}}
\newcommand{\cR}{\mathcal{R}}
\newcommand{\acyclic}{\mathsf{acyclic}}
\newcommand{\derivs}{\xRightarrow{*}}
\newcommand{\deriv}{\Rightarrow}

The key trick in our construction for OCA was that we could modify the subword construction so that the overapproximating NFA $\nb$ has the property that in any word from $\lang(\nb)$, we can repeat factors to obtain a word from $\na$. This was possible because in an OCA, essentially any pair of loops---one incrementing, one decrementing---could be repeated to pad a run.

However, in context-free languages, the situation is more complicated. With a
stack, any pumping must always ensure that stack contents match: It is not
possible to compensate stack effects with just two loops.  In terms of
grammars, the core idea for subword closures of context-free languages $L$ is
usually to overapproximate ``pump-like'' derivations $X\derivs uXv$ by
observing that---up to subwords---they can generate any $u'Xv'$ where the
letters of $u'$ can occur on the left and the letters of $v'$ can occur on the
right in derivations $X\derivs \cdot X\cdot$. Showing that all such words
belong to the downward closure leads to derivations $X\derivs
u''\bar{v}Xv''\bar{u}$, where $u'',v''$ are super-words of $u',v'$ such that
$X\derivs u''X\bar{u}$ and $X\derivs \bar{v}Xv''$ can be derived. The
additional infixes could introduce high priority letters and thus split blocks
unintentionally.

Therefore, we provide a novel recursive approach to compute the block downward
closure by decomposing derivations at high-priority letters. This is
non-trivial as this decomposition might not match the decomposition given
by derivation trees.  Formally, we show:
\begin{theorem}\label{main-cfl}
	Given a context-free language $L\subseteq\lowerletters{n}^*$, one can construct
	a doubly-exponential-sized automaton for $L\bd$, and thus also for
	$L\pd$.
\end{theorem}
We do not know if this doubly exponential upper bound is optimal. A
singly-exponential lower bound follows from the subword case: It is known that
subword downward closures of context-free languages can require exponentially
many states~\cite{DBLP:conf/lata/BachmeierLS15}.  However, it is not clear
whether for priority or block downward closures, there is a singly-exponential
construction.

We again note that \cref{generalizedToPriority} (and its proof) imply that for
\cref{main-cfl}, it suffices to compute a finite automaton for the block
downward closure of the context-free language: Computing the priority downward closure
then only increases the size polynomially.

\subparagraph{Grammars} We present the construction using \emph{context-free
grammars}, which are tuples $\cG=(N,T,P,S)$, where $N$ is a finite set of
\emph{non-terminal letters}, $T$ is a finite set of \emph{terminal letters}, $P$ is a
finite set of \emph{productions} of the form $X\to w$ with $X\in N$ and $w\in
(N\cup T)^*$, and $S$ is the \emph{start symbol}. For $u,v\in (N\cup T)^*$, we have $u\deriv
v$ if there is a production $X\to w$ in $P$ and $x,y\in (N\cup T)^*$ with
$u=xXy$ and $v=xwy$. The \emph{language generated by $\cG$}, is then
$\lang(\cG):=\{w\in T^* \mid S\derivs w\}$, where $\derivs$ is the reflexive,
transitive closure of $\deriv$.

\subparagraph{Assumption on the alphabet} In order to compute block downward
closures, it suffices to do this for flat alphabets (see
\cref{sec:block-order}).  The argument is essentially the same as in
\cref{generalizedToPriority}: By flattening the alphabet as in the proof of
\cref{generalizedToPriority}, we obtain a finer block order, so that first
computing an automaton for the flat alphabet and then applying
\cref{prioritydownwardregular} to the resulting finite automaton will yield a
finite automaton for the original (non-flat) alphabet.  In the following, we
will assume that the input grammar $\cG$ is in Chomsky normal form, meaning
every production is of the form $X\to YZ$ for non-terminals $X,Y,Z$, or of the
form $X\to a$ for a non-terminal $X$ and a terminal $a$.

\subparagraph{Kleene grammars} Suppose we are given a context-free grammar
$\cG=(N,\Sigma,P,S)$.  Roughly speaking, the idea is to construct another grammar
$\cG'$ whose language has the same block downward closure as $\lang(\cG)$, but
with the additional property that every word can be generated using a
derivation tree that is \emph{acyclic}, meaning that each path contains every
non-terminal at most once. Of course, if this were literally true, $\cG'$ would
generate a finite language.  Therefore, we allow a slightly expanded syntax: We
allow Kleene stars in context-free productions.

This means, we allow right-hand sides to contain occurrences of $B^*$, where
$B$ is a non-terminal. The semantics is the obvious one: When applying such a
rule, then instead of inserting $B^*$, we can generate any $B^k$ with $k\ge 0$.
We call grammars with such productions \emph{Kleene grammar}. A
\emph{derivation tree} in a Kleene grammar is defined as for context-free
grammars, aside from the expected modification: If some $B^*$ occurs on a
right-hand side, then we allow any (finite) number of $B$-labeled children in
the respective place. Then indeed, a Kleene grammar can generate infinite sets
using acyclic derivation trees. Given a Kleene grammar $\cH$, let
$\acyclic(\cH)$ be the set of words generated by $\cH$ using acyclic derivation
trees.
\begin{lemma}\label{acyclic-to-nfa}
	Given a Kleene grammar $\cH$, one can construct an exponential-sized
	finite automaton accepting $\acyclic(\cH)$.
\end{lemma}
\begin{proof}[Proof sketch]
	The automaton simulates a (say, preorder) traversal of an acyclic
	derivation tree of $\cH$.  This means, its state holds the path to the
	currently visited node in the derivation tree. Since every path has
	length at most $|N|$, where $N$ is the set of non-terminals of $\cH$,
	the automaton has at most exponentially many states.   
\end{proof}

Given \cref{acyclic-to-nfa}, for \cref{main-cfl}, it suffices to construct a
Kleene grammar $\cG'$ of exponential size such that
$\acyclic(\cG')\bd=\lang(\cG)\bd$.

\subparagraph{Normal form and grammar size} We will ensure that in the
constructed grammars, the productions are of the form (i)~$X\to w$, where $w$
is a word of length $\le 3$ and consisting of non-terminals $Y$ or Kleene stars
$Y^*$ or (ii)~$X\to a$ where $a$ is a terminal. This means, the total size of
the grammar is always polynomial in the number of non-terminals. Therefore, to
analyze the complexity, it will suffice to measure the number of non-terminals.
\newcommand{\lend}{\overleftarrow{\tau}}
\newcommand{\rend}{\overrightarrow{\tau}}
\newcommand{\rhash}{\overrightarrow{\#}}
\newcommand{\lhash}{\overleftarrow{\#}}

\newcommand{\Max}[1]{\Sigma_{\max #1}^+}
\newcommand{\MaxWithout}[1]{\Sigma_{\max #1}}
\subparagraph{Highest occurring priorities} Similar to classical downward
closure constructions for context-free languages, we want to overapproximate
the set of words generated by ``pump derivations'' of the form $X\derivs uXv$.
Since we are dealing with priorities, we first partition the set of such
derivations according to the highest occurring priorities, on the left and on
the right.  Thus, for $r,s\in[0,p]$, we will consider all derivations $X\derivs
uXv$ where $r$ is the highest occurring priority in $u$ and $s$ is the highest
occurring priority in $v$.  To ease notation, we define $\MaxWithout{r}$ to be
the set of words in $\lowerletters{r}^*$ in which $r$ is the highest occurring priority.
Since $\MaxWithout{r}=\Max{r}$, we will write $\Max{r}$ to remind us that this
is not an alphabet. Notice that for $r\in[1,p]$, we have
$\Max{r}=\lowerletters{r}^*r\lowerletters{r}^*$ and $\Max{0}=\lowerletters{0}^*$.

\subparagraph{Language of ends} 
In order to perform an inductive construction, we need
a way to transform pairs $(u,v)\in\Max{r}\times\Max{s}$ into words over an alphabet with fewer
priorities. Part of this will be achieved by the \emph{end maps}  $\lend_r(\cdot)$ and $\rend_s(\cdot)$ as follows.
Let $\hat{\Sigma}$ be the priority alphabet obtained from $\Sigma$ by adding the letters $\#$, $\lhash$, and $\rhash$ as letters with priority zero. Now for $r\in[1,p]$, the function $\lend_r\colon \Max{r}\to \hat{\Sigma}_{\leq r-1}^*$
is defined as:
\[ \lend_r(w) = u\lhash v,~\text{where $w=urx_1r\cdots x_nrv$ for some $n\ge 0$, $u,v,x_1,\ldots,x_n\in\lowerletters{r-1}^*$}. \]
Thus, $\lend_r(w)$ is obtained from $w$ by replacing the largest
possible infix surrounded by $r$ with $\lhash$. For $r=0$, it will be convenient to have the constant function $\lend_0\colon\Max{0}\to \{\lhash\}$.
Analogously, we define for $s\in[1,p]$
the function $\rend_s\colon \Max{s}\to\hat{\Sigma}_{\leq s-1}^*$ by
\[ \rend_s(w)=u\rhash v,~\text{where $w=usx_1s\cdots x_nsv$ for some $n\ge 0$, $u,v,x_1,\ldots,x_n\in\lowerletters{s-1}^*$}. \]
Moreover, we also set $\rend_0\colon\Max{0}\to\{\rhash\}$ to be the constant function yielding $\rhash$.

In particular, for $r,s\in[1,p]$, we have $\lend_r(w),\rend_s(w)\in\hat{\Sigma}_{\leq p-1}$
and thus we have reduced the number of priorities. Now consider for $r,s\in[0,p]$ the language
\begin{align*}
	E_{X,r,s} = \{ \lend_r(u)\#\rend_s(v) \mid X\derivs uXv,~u\in\lowerletters{r}^*r\lowerletters{r}^*,~v\in\lowerletters{s}^*s\lowerletters{s}^* \}.
\end{align*}
For the language $E_{X,r,s}$, it is easy to construct a context-free grammar:
\newcommand{\egrammar}{\cE}
\begin{restatable}{lemma}{transformEnd}\label{transform-end}
	Given $\cG$, a non-terminal $X$, and $r,s\in[0,p]$, one can construct a grammar $\egrammar_{X,r,s}$ for $E_{X,r,s}$ of linear size.
\end{restatable}

Defining the sets $E_{X,r,s}$ with fresh zero-priority letters $\#$, $\lhash$,
$\rhash$ is a key trick in our construction: Note that each word in $E_{X,r,s}$
is of the form $u\lhash v\# w\rhash x$ for $u,v,w,x\in\lowerletters{p-1}^*$. The
segments $u,v,w,x$ come from different blocks of the entire generated word, so
applying the block downward closure construction recursively to $E_{X,r,s}$
must guarantee that these segments embed as if they were blocks. However, there
are only a bounded number of segments. Thus, we can reduce the number of
priorities while retaining the block behavior by using fresh zero-priority
letters.  This is formalized in the following
\lcnamecref{embedding-fresh-letter}:
\begin{restatable}{lemma}{embeddingFreshLetter}\label{embedding-fresh-letter}
For $u,u',v,v'\in\lowerletters{p}^*$, we have $u\#v\bo u'\#v'$ iff both (i)~$u\bo u'$ and (ii)~$v\bo v'$.
\end{restatable}

\subparagraph{Language of repeated words} Roughly speaking, the language
$E_{X,r,s}$ captures the ``ends'' of words derived in derivations $X\derivs
uXv$ with $u\in\Max{r}$ and $v\in\Max{s}$: On the left, it keeps everything
that is not between two occurrences of $r$ and on the right, it keeps
everything not between two occurrences of $s$. We now need languages that
capture the infixes that can occur between $r$'s and $s$'s, respectively.
Intuitively, these are the words that can occur again and again in words
derived from $X$. There is a ``left version'' and a ``right version''. We set
for $r,s\in[1,p]$:
\newcommand{\lrepeat}{\overleftarrow{R}}
\newcommand{\rrepeat}{\overrightarrow{R}}
\newcommand{\lrgrammar}{\overleftarrow{\cR}}
\newcommand{\rrgrammar}{\overrightarrow{\cR}}
\begin{align*}
	\lrepeat_{X,r,s} &= \{yr \mid y\in\lowerletters{r-1}^*,~\exists x,z\in\lowerletters{r}^*,~v\in\Max{s}\colon X\derivs xryrzXv \} \\
	\rrepeat_{X,r,s} &= \{ys \mid y\in\lowerletters{s-1}^*,~\exists u\in\Max{r},~x,z\in\lowerletters{r}^*\colon X\derivs uXxsysz \}.
\end{align*}
The case where one side has highest priority zero must be treated slightly differently: There are no enveloping occurrences of some $r,s\in[1,p]$. However, we can overapproximate those words by the set of all words over a particular alphabet. Specifically, for $r,s\in[0,p]$, we set
\begin{align*}
	\rrepeat_{X,0,s} &= \{a\in\lowerletters{0} \mid \exists u\in\Max{0},~v\in\Max{s}\colon X\derivs uXv,~\text{$a$ occurs in $u$} \} \\
	\lrepeat_{X,r,0} &= \{a\in\lowerletters{0} \mid \exists u\in\Max{r},~v\in\Max{0}\colon X\derivs uXv,~\text{$a$ occurs in $v$} \}
\end{align*}

\begin{restatable}{lemma}{transformRepeat}\label{transform-repeat}
	Given $\cG$, a non-terminal $X$, and $r,s\in[0,p]$, one can construct grammars
	$\lrgrammar_{X,r,s}$, $\rrgrammar_{X,r,s}$ for $\lrepeat_{X,r,s}$,$\rrepeat_{X,r,s}$,
	respectively, of linear size.
\end{restatable}

\subparagraph{Overapproximating derivable words}
The languages $E_{X,r,s}$ and $\lrepeat_{X,r,s}$ and $\rrepeat_{X,r,s}$ now serve to define overapproximations of the set of $(u,v)\in\Max{r}\times\Max{s}$ with $X\derivs uXv$: One can obtain each such pair by taking a word from $E_{X,r,s}$, replacing $\lhash$ and $\rhash$, resp., by words in $r\lrepeat_{X,r,s}^*$ ($\lrepeat_{X,0,s}^*$ if $r=0$) and $s\rrepeat_{X,r,s}^*$ ($\rrepeat_{X,r,0}^*$ if $s=0$), respectively. By choosing the right words from $E_{X,r,s}$, $\lrepeat_{X,r,s}$, and $\rrepeat_{X,r,s}$, we can thus obtain $u\# v$. However, this process will also yield other words that cannot be derived. However, the key idea in our construction is that every word obtainable in this way from $E_{X,r,s}$, $\lrepeat_{X,r,s}$, and $\rrepeat_{X,r,s}$ will be in the block downward closure of a pair of words derivable using $X\derivs \cdot X\cdot$.

Let us make this precise. To describe the set of words obtained from
$E_{X,r,s}$, $\lrepeat_{X,r,s}$, and $\rrepeat_{X,r,s}$, we need the notion of
a substitution. For alphabets $\Gamma_1,\Gamma_2$, a \emph{substitution}  is a
map $\sigma\colon\Gamma_1\to2^{\Gamma_2^*}$ that yields a language in
$\Gamma_2$ for each letter in $\Gamma_1$. Given a word $w=w_1\cdots w_n$ with
$w_1,\ldots w_n\in\Gamma_1$, we define
$\sigma(w):=\sigma(w_1)\cdots\sigma(w_n)$. Then for $K\subseteq\Gamma_1^*$, we
set $\sigma(K)=\bigcup_{w\in K}\sigma(w)$. Now let
$\Sigma_{X,r,s}\colon\hat{\Sigma}_{\leq p}\to 2^{\hat{\Sigma}_{\leq p}^*}$ be the
substitution that maps every letter in $\lowerletters{p}\cup\{\#\}$ to itself (as a
singleton) and maps $\lhash$ to $r\lrepeat_{X,r,s}^*$ and $\rhash$ to
$s\rrepeat_{X,r,s}^*$. Now our observation from the previous paragraph can be phrased as:
\begin{restatable}{lemma}{grammarCorrectness}\label{grammar-correctness}
	For every $u\#v\in \Sigma_{X,r,s}(E_{X,r,s})$, there are $u'\in\Max{r}$ and $v'\in\Max{s}$ with $u\bo u'$, $v\bo v'$, and $X\derivs u'Xv'$.
\end{restatable}

\subparagraph{Constructing the Kleene grammar} We now construct the Kleene
grammar for $\lang(\cG)\bd$ by first computing the grammars
$\egrammar_{X,r,s}$, $\lrgrammar_{X,r,s}$, and $\rrgrammar_{X,r,s}$ for each non-terminal $X$ and each $r,s\in[1,p]$.
Then, since $\egrammar_{X,r,s}$, $\lrgrammar_{X,r,s}$, and $\rrgrammar_{X,r,s}$ generate languages with
at most $p-1$ priorities, we can call our construction recursively to 
obtain grammars $\egrammar'_{X,r,s}$, $\lrgrammar'_{X,r,s}$, and $\rrgrammar'_{X,r,s}$, respectively.
Then, we add all productions of the grammars $\egrammar'_{X,r,s}$, $\lrgrammar'_{X,r,s}$, and
$\rrgrammar'_{X,r,s}$ to $\cG'$.  Moreover, we make the following
modifications: Each production of the form $Y\to \lhash$ (resp. $Y\to \rhash$)
in $\egrammar_{X,r,s}$ is replaced with $Y\to Z_r\overleftarrow{S}_{X,r,s}^*$ (resp.\
$Y\to Z_s\overrightarrow{S}_{X,r,s}^*$), where $\overleftarrow{S}_{X,r,s}$ (resp.\
$\overrightarrow{S}_{X,r,s}$) is the start symbol of $\lrgrammar'_{X,r,s}$ (resp.
$\rrgrammar'_{X,r,s}$), and $Z_r$ is a fresh non-terminal used to derive $r$ or $\varepsilon$: We also have $Z_r\to r$ for each $r\in[1,p]$ and $Z_0\to\varepsilon$. 
Moreover, each production $Y\to\#$ in $\egrammar'_X$ is
removed and replaced with a production $Y\to w$ for each production $X\to w$ in $\cG$.
We call the resulting grammar $\cG'$.

\subparagraph{Correctness} Let us now observe that the grammar $\cG'$ does
indeed satisfy $\lang(\cG')\bd=\lang(\cG)\bd$. The
inclusion ``$\supseteq$'' is trivial as $\cG'$ is obtained by adding
productions. For the converse, we need some terminology. We say that a
derivation tree $t_1$ in $\cG'$ is obtained using an \emph{expansion step} from
$t_0$ if we take an $X$-labeled node $x$ in $t_0$, where $X$ is a non-terminal
from $\cG$, and replace this node by a derivation $X\derivs uwv$ using newly
added productions (i.e. using $\egrammar_{X,r,s}$, $\lrgrammar_{X,r,s}$, and $\rrgrammar_{X,r,s}$
and some $Y\to w$ where $X\to w$ was the production applied to $x$ in $t_0$).
Then by construction of $\cG'$, any derivation in $\cG'$ can be
obtained from a derivation in $\cG$ by finitely many expansion steps. An induction on the number of expansion steps shows: 
\begin{restatable}{lemma}{cflCorrectness}\label{cfl-correctness}
We have $\lang(\cG')\bd=\lang(\cG)\bd$.
\end{restatable}

\subparagraph{Acyclic derivations suffice} Now that we have the grammar $\cG'$
with $\lang(\cG')\bd=\lang(\cG)\bd$, it remains to show
that every word in $\cG'$ can be derived using an acyclic derivation:
\begin{restatable}{lemma}{cflAcyclicity}\label{cfl-acyclicity}
$\acyclic(\cG')\bd=\lang(\cG)\bd$.
\end{restatable}
Essentially, this is due to the fact that any repetition of a non-terminal $X$
on some path means that we can replace a corresponding derivation $X\derivs
uXv$ by using new productions from $\egrammar'_{X,r,s}$, $\lrgrammar'_{X,r,s}$, and
$\rrgrammar'_{X,r,s}$. Since these also have the property that every derivation can
be made acyclic, the lemma follows. See \cref{appendix-cfl} for details.

\subparagraph{Complexity analysis} To estimate the size of the constructed grammar,  let $f_p(n)$ be the maximal number of non-terminals of
a constructed Kleene grammar for an input grammar with $n$ non-terminals over
$p$ priorities. By \cref{transform-end,transform-repeat}, there is a constant $c$ such that each grammar $\egrammar_X$, $\lrgrammar_X$, and $\rrgrammar_X$ has at most $cn$ non-terminals.
Furthermore, $\cG'$ is obtained by
applying our construction to $3n(p+1)^2$ grammars with $p-1$ priorities of
size $cn$, and adding $Z_p$. Thus $f_p(n)\le n+3n(p+1)^2f_{p-1}(cn)+1$.
Since $f_{p-1}(n)\ge 1$, we can simplify to $f_p(n)\le 4n(p+1)^2 f_{p-1}(cn)$.
It is easy to check that $f_0(n)\le 4n+1\le 5n$, because
$\egrammar_{X,0,0}$ and $\lrgrammar_{X,0,0}$ and $\rrgrammar_{X,0,0}$ each only
have one non-terminal.  Hence $f_p(n)\le (4n(p+1)^2)^pf_0(c^pn)\le
(4n(p+1)^2)\cdot 4(c^pn)$, which is exponential in the size of $\cG$.

\section{Conclusion}\label{sec:conclusion}
We have initiated the study of computing priority and block downward closures for infinite-state systems. We have shown that for OCA, both closures can be computed in polynomial time. For CFL, we have provided a doubly exponential construction.

Many questions remain. First, we leave open whether the doubly exponential bound for context-free languages can be improved to exponential. An exponential lower bound is easily inherited from the exponential lower bound for subwords~~\cite{DBLP:conf/lata/BachmeierLS15}. Moreover, it is an intriguing question whether computability of subword downward closures for vector addition systems~\cite{DBLP:conf/icalp/HabermehlMW10}, higher-order pushdown automata~\cite{DBLP:conf/popl/HagueKO16}, and higher-order recursion schemes~\cite{DBLP:conf/lics/ClementePSW16} can be strengthend to block and priority downward closures.

\bibliography{main}

\appendix

\section{Missing proofs from~\cref{sec:block-order}}\label{app:block}

Before we prove \cref{thm:generalizedblockwqo}, we recall some well-known facts that are used in the proof. WQOs are preserved under many operations on quasi-ordered sets. We mention some of these operations below, which will be used later in the paper to show the WQO property. The reader is referred to Halfon's thesis \cite{halfon:tel-01945232} for a well informed survey of such results. We mention some of these results below.

\begin{lemma}[Dickson's Lemma, \cite{10.2307/2370405} ]
	Given two {\wqo}s $ (A,\leq_1) $ and $ (B,\leq_2) $, its product $ (A\times B,\leq_1\times \leq_2) $ is also a \wqo. Here, $ (a_1,b_1)\leq_1\times \leq_2 (a_2,b_2) $ if $ a_1\leq_1 a_2 $ and $ b_1\leq_2 b_2 $.
\end{lemma}
\begin{lemma}[Higman's Lemma, \cite{10.1112/plms/s3-2.1.326}]\label{higmanLemma}
	$ (X^*,\leq_*) $ is a \wqo $ \iff $ $ (X,\leq) $ is a \wqo. We say $ a_1a_2\cdots a_k \leq_* b_1b_2\cdots b_l  $ if there is a strictly monotonically increasing map $ \phi:[1,k]\rightarrow [1,l] $ such that $ \forall\ i\in [1,k], a_i\leq b_{\phi(i)} $.
\end{lemma}
\begin{lemma}[Monomorphism Lemma, \cite{Kruskal1960WellquasiorderingTT}]
	Let $ (X,\leq_1) $ and $ (Y,\leq_2) $ be two quasi orders, and $ h:X\rightarrow Y $ be a monomorphism. If $ (Y,\leq_2) $ is a \wqo then $ (X,\leq_1) $ is a \wqo.
\end{lemma}

We prove that the block order is a WQO. We restate \cref{thm:generalizedblockwqo} for convenience of the reader.
\thmgeneralizedblockwqo*
\begin{proof}
	We will prove the lemma by induction on the size of $ \priority $. Firstly, we note that generalized block order and subword order coincide for singleton priority set i.e. $ u\bo v \iff u\so v $, by definition. Since $ (\Sigma^*,\so) $ is a \wqo, this gives us the base case, i.e. if $ \priority $ is singleton, then $ (\Sigma^*,\bo) $ is a \wqo.
	
	Now, for the induction hypothesis, assume that the lemma is true for $ \priority=[0,p-1] $. We show that the lemma holds for $ \priority=[0,p] $.
	
	Since $ (A,=) $ is a \wqo for any finite $ A $, by Dickson's lemma, $ (\Sigma_{p-1}^*\times A_p,~\bo\times =) $ is a \wqo. Then, by Higman's lemma, $ ((\Sigma_{p-1}^*\times A_p)^*,\ (\bo\times =)_*) $ is a \wqo. Again, by Dickson's lemma, $ ((\Sigma_{p-1}^*\times A_p)\times(\Sigma_{p-1}^*\times A_p)^*\times \Sigma_{p-1}^*,\ (\bo\times =)\times(\bo\times =)_*\times \bo) $ is a \wqo.
	
	Now, consider the function  \[ h: (\Sigma_{p}^*,\bo) \rightarrow ((\Sigma_{p-1}^*\times A_p)\times(\Sigma_{p-1}^*\times A_p)^*\times \Sigma_{p-1}^*,\ (\bo\times =)\times(\bo\times =)_*\times \bo)  \]  defined as,
	\[ u_0y_0u_1y_1u_2y_2\cdots y_{k-1}u_k\mapsto ((u_0,y_0),(u_1,y_1),(u_2,y_2),\ldots ,(u_{k-1},y_{k-1}), u_k),  \]
	where $ u_i $s are sub-$ p $ blocks.
	
	It is easy to see that $ h $ is a monomorphism. Then, by the monomorphism lemma, we get that $ (\Sigma_{p}^*,\bo) $ is a \wqo.
\end{proof}

We now introduce the notion of upward closed sets, which allows us to prove the \cref{generalizedblockregular}.

\noindent
\textbf{Upward closure.} Upward closure is the dual of downward closure. Given a set $ S $ with a partial order $ \triangleleft $, the $ \triangleleft $-upward closure of $ L\subseteq S $, denoted by $ L\uparrow_{\triangleleft} $, is the set of elements of $ S $ which are larger \wrt $ \triangleleft $ than some element in $ L $, i.e.
\[ L\uparrow_{\triangleleft} = \{ u\in S\, |\ \exists\ v \text{ such that } v\ \triangleleft\ u \} \]
A subset $ L $ of $ S $ is called \textit{$ \triangleleft $-upward closed} if $ L=L\uparrow_{\triangleleft} $.

The \textit{subword upward closure} and \textit{block upward closure} are defined by taking the set of finite words $ \Sigma^* $ with the partial orders $ \so $ and $ \bo $, respectively. For $ L\subseteq\Sigma^* $,
\begin{eqnarray*}
	L\uparrow &=&\{ u\in S\, |\ \exists\ v \text{ such that } v\ \so\ u  \} \\
	L\Uparrow &=&\{ u\in S\, |\ \exists\ v \text{ such that } v\ \bo\ u  \} 
\end{eqnarray*}
It is easy to see that the complement of a downward closed set is an upward closed set and vice versa, hence they are dual of each other. %
The following theorem characterizes regular sets using upward closure.

\begin{theorem}[\cite{EHRENFEUCHT1983311}]\label{multiplicativewqoregular}
	A set $ S $ is regular iff it is the upward closure of some multiplicative \wqo.
\end{theorem}

We restate the \cref{subwordregular} and \cref{generalizedblockregular} below. %

\donwardclosuresregular*
The proof for priority downward closure is analogous for that of block order, as shown below.
\lemgeneralizedblockregular*
\begin{proof}	
	Since the complement of the downward closed set is an upward closed set, regular languages are closed under complementation, and block order is a multiplicative WQO (\cref{generalizedblockmultiplicative}), the proof of the lemma is a simple corollary of \cref{multiplicativewqoregular}.
\end{proof}

\section{Missing proofs from~\cref{sec:regular}}\label{app:regular}

\sizeprioritytrans*
\begin{proof}
	For the block order consider the transducer that has one state for every priority and a sink state, and for every state it reads a letter, and
	\begin{itemize}
		\item if the letter has lower or equal priority as the state, does not output it, and stays there,
		\item if the letter has equal priority, outputs it, and goes to state $ 0 $, and
		\item for other scenarios, goes to the sink state.
	\end{itemize}
	At state $ 0 $, if the letter is output, then it stays at $ 0 $, else goes to the state with priority of the letter. This intuitively allows dropping whole sub-$ i $ blocks until priority $ i $ is output again. The construction exploits the fact that between two consecutive letters which are not dropped, no bigger priority letter is dropped.
	
	Similarly, for priority order, we have the same state space, along with another accepting state.
	\begin{itemize}
		\item if the letter has strictly lower priority than the state, then does not output it, and stays there,
		\item if the letter has same or higher priority, does not output it, and goes to the state with priority of the letter,
		\item if the letter has same or higher priority, outputs it, and goes to the state with priority 0, or the accepting state non-deterministically.
	\end{itemize}
	The priority $ 0 $ state is the initial state, and the new accepting state is the final state. Intuitively, the transducer remembers the largest priority letter that has been dropped, and keeps only a letter of higher priority later. To be accepting, it has to read the last letter to go to the accepting final state. 
\end{proof}

\section{Missing proofs from~\cref{sec:oca}}\label{app:oca}

We restate the \cref{blockdownwardsimpleOCA} below and prove it formally.

\blockdownwardsimpleOCA*

\begin{proof}
	In this proof we use the shorthand $ [n] $ for $ [0,n]= \{0,\ldots, n\} $, and $ (n) $ for $ [1,n]=\{1,\ldots,n\} $.
	We describe the construction of the intermediate \FSA $ \nb $ formally. Let $ U\geq K^2+K+1 $.
	\[ \nb=(Q_1\cup Q_2\cup Q_3,\Sigma,\Delta,q_0',F') \]
	where $ Q_1=Q\times[K]\times\{1\}, Q_2=(Q\times[U]\times\{2\})\cup (Q\times Q) $ and $ Q_3=Q\times[K]\times\{3\} $. We let $ q_0'=(q_0,0,1) $ and $ F'=\{(q_f,0,1),(q_f,0,3) \} $. The transition relation is the union of the relations $ \Delta_1, \Delta_2 $ and $ \Delta_3 $ defined as follows:
	
	\noindent
	\textbf{Transitions in $ \Delta_1 $:}
	\begin{enumerate}
		\setlength\itemsep{0em}
		\item $ (q,n,1)\xrightarrow{a}(q',n,1) $ for all $ n\in[K] $ whenever $ (q,a,i,q')\in\delta $. Simulate an internal move.
		\item $ (q,n,1)\xrightarrow{a}(q',n-1,1) $ for all $ n\in(K) $ whenever $ (q,a,-1,q')\in\delta $. Simulate a decrement.
		\item $ (q,n,1)\xrightarrow{a}(q',n+1,1) $ for all $ n\in[K-1] $ whenever $ (q,a,+1,q')\in\delta $. Simulate an increment.
		\item $ (q,K,1)\xrightarrow{a}(q',K+1,2) $ whenever $ (q,a,+1,q')\in\delta $. Simulate an increment and shift to second phase.
	\end{enumerate}	
	\textbf{Transitions in $ \Delta_2 $:}
	\begin{enumerate}
		\setlength\itemsep{0em}
		\item $ (q,n,2)\xrightarrow{a}(q',n,2) $ for all $ n\in[U] $ whenever $ (q,a,i,q')\in\delta $. Simulate an internal move.
		\item $ (q,n,2)\xrightarrow{a}(q',n-1,2) $ for all $ n\in(U) $ whenever $ (q,a,-1,q')\in\delta $. Simulate a decrement.
		\item $ (q,n,2)\xrightarrow{a}(q',n+1,2) $ for all $ n\in[U-1] $ whenever $ (q,a,+1,q')\in\delta $. Simulate an increment.
		\item $ (q,K+1,1)\xrightarrow{a}(q',K,3) $ whenever $ (q,a,+1,q')\in\delta $. Simulate an decrement and shift to third phase.
		\item $ (q,n,2)\xrightarrow{ \epsilon }(q,q) $. Start simulating OCA as an NFA.
		\item $ (q,q)\xrightarrow{ \epsilon }(q,n,2) $. Stop simulating OCA as an NFA.
		\item $ (q_1,q_2)\xrightarrow{a}(q_1',q_2) $ for all $ a\in\Sigma $ and $ q_2\in Q $ whenever $ (q_1,a,x,q_1')\in\delta $ for some $ x\in \{+1,-1,0,z  \} $. Simulate OCA as an \FSA starting from $ q_2 $.
	\end{enumerate}
	\textbf{Transitions in $ \Delta_3 $:}
	\begin{enumerate}
		\setlength\itemsep{0em}
		\item $ (q,n,3)\xrightarrow{a}(q',n,3) $ for all $ n\in[K] $ whenever $ (q,a,i,q')\in\delta $. Simulate an internal move.
		\item $ (q,n,3)\xrightarrow{a}(q',n-1,3) $ for all $ n\in(K) $ whenever $ (q,a,-1,q')\in\delta $. Simulate a decrement.
		\item $ (q,n,3)\xrightarrow{a}(q',n+1,3) $ for all $ n\in[K-1] $ whenever $ (q,a,+1,q')\in\delta $. Simulate an increment.
	\end{enumerate}
	
	\ocacontainsoriginaloca*
	\begin{claimproof}
		Let $ w\in\lna $. Then there is a run $ \rho $ in $ \lna $ on $ w $, which can be partitioned as follows,
		\[ \rho = (q_0,0)\xrightarrow{x}(q_1,s)\xrightarrow{y}(q_2,t)\xrightarrow{z}(f,0) \]
		where $ \rho_1= (q_0,0)\xrightarrow{x}(q_1,s) $ is the longest prefix such that the counter value stays below $ K $, and $ \rho_3=(q_2,t)\xrightarrow{z}(f,0) $ is the longest suffix disjoint from $ \rho_1 $ such that the counter value stays below $ K $, and $ \rho_2=(q_1,s)\xrightarrow{y}(q_2,t) $. Since $ Q_1 $ and $ Q_2 $ can simulate $ \na $ by keeping track of counter values below $ K $, we know that there are runs $ \rho'_1= (q_0,0,1)\xrightarrow{x}(q_1,s,1) $ and $ \rho'_3=(q_2,t,1)\xrightarrow{z}(q_3,0,1) $ in $ \nb $. We also observe that if the counter value does not go above $ K $ in $\rho$, then $ \rho_2 $ and $ \rho_3 $ are empty, and $ \rho=(q_0,0)\xrightarrow{w}(f,0) $. So $ (q_0,0,1)\xrightarrow{w}(f,0,1) $ is a valid and accepting run in $ \nb $.
		
		So now suppose $ \rho $ exceeds $ K $ in the counter. Now if the counter value stays below $ U $, then $ \rho_2 $ can be simulated by $ Q_2 $ and it's transitions, and \[ \rho'= (q_0,0,1)\xrightarrow{x}(q_1,K,1)\xrightarrow{y}(q_2,K,1)\xrightarrow{z}(f,0,1) \]
		is a valid run in $ \nb $.
		
		Let the maximum counter value reached in the run $ \rho $ be $ m $. If $ m\geq U $, then we show that the run can be shortened to keep the counter value below $ U $, and the new run along with the trimmed part can be simulated by $ \nb $. 
		
		Let $ \rho_m $ be the shortest prefix of $ \rho_2 $ such that at the end of $ \rho_m $ the counter value is $ m $. For each $ K\leq i\leq m $, consider the pair of states $ (p_l^i, p_r^i) $, such that $ (p_l^i,i) $ be the last configuration in $ \rho_m $ with the counter value $ i $, and $ (p_r^i,i) $ be the first configuration in $ \rho_2 $ after $ \rho_m $ such that the counter value is $ i $. Since we have only $ K $ many vertices, we have $ K^2 $ such pairs. But $ i $ ranges from $ K+1 $ to $ K^2+K+1 $, by PHP, we have that for some $ K\leq i<j\leq m $, $ (p_l^i, p_r^i) =(p_l^j, p_r^j)  $. Then the runs $ (p_l^i,i)\xrightarrow{y_{(i,j)}} (p_l^j,j) $ $ (p_r^j,j)\xrightarrow{y_{(j,i)}} (p_r^i,i) $ can be removed from $ \rho_2 $ to reduce the counter value, and can be simulated in $ \nb $ by edges of type 5, 6, and 7 in $ \Delta_2 $. We do this repeatedly to get a shorter run in $ \na $ which does not exceed $ U $, and simulate the trimmed parts by edges of type $ 5,6, $ and $ 7 $. And this shorter run can be trivially simulated by edges of type $ 1-4 $ in $ \Delta_2 $. Hence, we can simulate $ \rho $ in $ \nb $, and $ w\in\lnb $.%
	\end{claimproof}
	
	\ocadownwardclosurecontainsoca*
	\begin{claimproof}
		Let $ w\in\lnb $, and let $ \rho $ be the witnessing run. Let the minimum value of the counter in $ \rho $ be $ m $. If $ m\geq 0 $, then it is a run in $ \na $, and there is nothing to show. Now suppose $ m $ is negative. Then let 
		\[ \rho=(q_0,0,1)\xrightarrow{x}(q_1,K,1)\xrightarrow{a}(q_2,K+1,2)\xrightarrow{y}(q_3,K+1,2)\xrightarrow{b}(q_4,K,1)\xrightarrow{z}(f,0,1). \]
		For $ 0\leq i\leq K $, consider the state $ p_i $ such that $ (p_i,i,1) $ is the first configuration along $ (q_0,0,1)\xrightarrow{x}(q_1,K,1) $ with the counter value $ i $. Then by PHP, there exist $ 0\leq i<j\leq K $, such that $ p_i=p_j $. Let the run between $ (p_i,i) $ to $ (p_j,j) $ be $ \rho_l $, and the counter difference $ k_1=j-i $. Similarly, there exist a similar run $ \rho_r $ with  counter difference $ k_2 $ in $ (q_4,K,3)\xrightarrow{z}(f,0,3) $. Notice that $ \rho_l $ can be pumped to make the counter value arbitrary high, and similarly, $ \rho_r $ can be pumped to bring down the counter value from arbitrary high value. Moreover, observe that since $ m $ is negative, there must exist
		\[ \rho_c=(q,n,2)\xrightarrow{\epsilon} (q,q)\xrightarrow{u}(q,q)\xrightarrow{\epsilon}(q,n,2), \]
		such that the counter value reduces by $ m'>0 $ after this execution.
		
		Now consider a $ N $ such that $ k_1N + m>0 $. Then on pumping $ \rho_l $ $ k_2Nm' $ times, the counter value before $ \rho_c $ becomes $ n+k_1k_2Nm' $. And executing $ u $ $ k_2m $ times makes the counter value $ n+k_1k_2Nm'+k_2mm'= n+ (k_1N+m)k_2m'>0 $. Then pumping $ \rho_r $ $ (k_1N+m)m' $ times brings the counter value to 0 in the end. 
		
		However, this will give a run on word of the form $ w'=x_1x_2^{K_1}x_3ay_1u^{K_2}y_2bz_1z_2^{K^3}z_3 $  where $ K_1= k_2Nm' $, $ K_2 = k_2m $ and $ K_3 = (k_1N+m)m' $, such that $ x_1x_2x_3=x $, $ y_1uy_2=y $ and $ z_1z_2z_3=z$. But from lemma \ref{generalizedblockrepeat}, we know that $ w\bo w'  $. Since $ w'\in \lna $, $ w\in\lna\bd $.
	\end{claimproof}
	With this we have shown that $ \nb $ has the same downward closure as $ \na $. And observe that the $ \nb $ is a NFA with polynomially many states, $ K^3+3K^2+K $, where $ K= |Q| $.
\end{proof}

\section{Missing proofs from~\cref{sec:cfl}}\label{appendix-cfl}
\begin{lemma}\label{X-to-hash}
	Given a non-terminal $X$ in a context-free grammar, one can construct a
	linear-size grammar for the language $\{u\#v\mid X\derivs uXv \}$.
\end{lemma}

\begin{lemma}\label{compute-side-alphabets}
	Given a non-terminal $X$, one can compute in polynomial time the
	alphabets $\overleftarrow{\Gamma}_X$ and $\overrightarrow{\Gamma}_X$.
\end{lemma}
\begin{proof}
	First, apply \cref{X-to-hash} to construct a grammar for $K=\{u\#v\mid
	X\derivs uXv \}$.  Then, we can decide whether
	$a\in\overleftarrow{\Gamma}_X$ by checking whether $K$ intersects the
	regular language $R_a=\{xay\#v\mid x,y,v\in\Sigma_0^*\}$, for which one
	can construct a three-state automaton. Since intersection emptiness
	between a context-free language and a given regular language is
	decidable in polynomial time, the set $\overleftarrow{\Gamma}_X$ can be
	computed in polynomial time. An analogous argument holds for
	$\overrightarrow{\Gamma}_X$.
\end{proof}

Since the context-free languages are closed under rational transduction, we use the standard triple construction (see e.g., \cite[Prop. 2.6.1]{Zetzsche2016c}) to obtain new grammars after applying transductions. The technique allows construction of new grammar of size linear in the original grammar and polynomial in the size of the transducer. 
\begin{lemma}\label{cfg-transducers}
	Given a CFG $ \cG $ recognizing a language $ L $, and a transducer $ \nb $ defining a transduction $ T $, the language $ TL $ is recognized by a CFG of size $ |\cG|\cdot |Q|^2 $, where $ Q $ is the number of states in $ \nb $.
\end{lemma}

\subsection{Proof of \cref{transform-end}}
\transformEnd*
\begin{proof}
	We first apply \cref{X-to-hash} to construct a grammar $ \cG' $ for $ K=\{u\#v\mid X\derivs uXv \} $. Then consider the following transducer $ T $, that first reads and outputs all the letters until an $ r $ is seen. Once an $ r $ is read, it outputs $ \lhash $ and keeps dropping subsequent letters until on reading another $ r $ it non-deterministically decides that $ r $ will not be seen before $ \# $. Then if another $ r $ is read before the $ \# $, it rejects the run by going to a non-final sink state. Otherwise, it outputs the letters that are seen until it encounters a $ \# $, and outputs it. Then it reads and outputs all the letters until a $ s $ is read, in which case it outputs $ \rhash $, and continues dropping the letters until another $ s $ is read, and it non-deterministically decides to not read a $ s $. It output all the letters after this $ s $. Moreover, it goes to the sink state if it reads a letter with priority greater than $ r $ (and, priority greater than $ s $) before the $ \# $ (after the $ \# $).
	
	Note that the transducer only needs $ 7 $ states; one reject state, 3 states for the right part of $ \# $, and 3 for the left part. We apply the transducer $ T $ to the grammar $ \cG' $ to obtain a grammar of size $ 49|\cG'| $, which is linear in the original grammar, due to \cref{X-to-hash}.
	
	Although we show the transducer for $ r,s>0 $, for the case of $ r=0 $ or $ s=0 $, the transducer just outputs $ \lhash $ or $ \rhash $ accordingly.  
\end{proof}

\subsection{Proof of \cref{embedding-fresh-letter}}
\embeddingFreshLetter*
\begin{proof}
For the forward direction, let us assume that $ u\#v\bo u'\#v' $. Then suppose the largest priority occurring in $ u\#v $ and $ u'\#v' $ be $ p $. Then there exists a witness block map $ \rho $ for $ u\#v\bo u'\#v' $.

Now let $ \# $ belongs to the $ m^{th} $ and $ n^{th} $ sub-$ p $ blocks of $ u\#v $ and $ u'\#v' $ respectively. We then show that $ u\bo u' $. Consider the block map $ \rho_u $ that maps $ i^{th} $ sub-$ p $ block of $ u $ to the $ \rho(i)^{th} $ sub-$ p $ block of $ u' $ for all $ i\in [0,m] $. By definition of $ \rho $, $ u_i\bo u'_{\rho_u(i)} $, where $ w_i $ denotes the $ i^{th} $ sub-$ p $ block of $ w $ for all $ i\in[0,n-1] $. Moreover, $ \rho_u(0) = 0 $, and $ \rho_u(m)=n $. It only remains to show that $ u_m\bo u'_{n} $. But this holds recursively, since the $ 0 $-block that $ \# $ of $ u $ belongs to is subword smaller than the $ 0 $-block that $ \# $ of $ u' $ belongs to. 

A similar argument shows that the block map $ \rho_v $ that maps $ i^{th} $ sub-$ p $ block of $ v $ to $ (\rho(i+m)-n)^{th} $ sub-$ p $ block of $ v' $ is the required witness block map. 

Now, for the other direction, let $ \rho_u $ and $ \rho_v $ be the witness block maps for $ u\bo u' $ and $ v\bo v' $. Then consider the block map $ \rho $ such that 
\[ i\mapsto 
	\begin{cases}
		\rho_u(i), \text{ if } i\geq n\\
		\rho_v(i-m)+n, \text{ otherwise}.
	\end{cases} \]
Again it suffices to show that the $ m^{th} $ sub-$ p $ block of $ u\#v $ is block smaller than the $ n^{th} $ sub-$ p $ block of $ u'\#v' $. But this again recursively holds since the $ 0 $-block that $ \# $ of $ u\#v $ belongs to is subword smaller than the $ 0 $-block that $ \# $ of $ u'\#v' $ belongs to. 
\end{proof}

\subsection{Proof of \cref{transform-repeat}}
\transformRepeat*
\begin{proof}
	Due to \cref{cfg-transducers}, we again only construct a transducer of constant size that results in $ \lrepeat_{X,r,s} $ when applied to $ \cG' $ where $ \cG' $ is the grammar for the language $ K= \{u\#v\mid X\derivs uXv\} $ obtained via \cref{X-to-hash}. The case of $ \rrepeat_{X,r,s} $ is analogous. 
	
	Consider the following transducer $ T $. The transducer has a non-final sink state, which we call the rejecting state. The transducer reads the letters (with equal or less priority than r) and does not output anything (i.e. outputs $ \epsilon $), until it reads an $ r $ and decides to output the next sub-$ r $ block non-deterministically. It then outputs all the letters read till the next $ r $, and then does not output the subsequent letters. On reading $ \# $, it outputs nothing, but verifies if the highest occurring letter in the right of $ \# $ is $ s $. If that is the case, it accepts, otherwise rejects. It is clear from the construction that $ T\lang(\cG') $ is the language $ \lrepeat_{X,r,s} $.
	
	Note that the transducer has 5 states: one rejecting state, 3 to output sub-$ r $ block on the left of $ \# $, and one to verify if the word on the right of $ \# $ is in $ \Max{s} $.
	
	The cases when $ r=0 $ or $ s=0 $ are rather straightforward, as the transducer just non-deterministically outputs one arbitrary letter from the corresponding side of $ \# $.
\end{proof}

\subsection{Proof of \cref{grammar-correctness}}
\grammarCorrectness*
\begin{proof}
	For every $u\#v\in \sigma_{X,r,s}(E_{X,r,s})$, there are $u'\in\Max{r}$ and $v'\in\Max{s}$ with $u\bo u'$, $v\bo v'$, and $X\derivs u'Xv'$.

\newcommand{\lw}{\overleftarrow{w}}
\newcommand{\rw}{\overrightarrow{w}}
\newcommand{\lrho}{\overleftarrow{\rho}}
\newcommand{\rrho}{\overrightarrow{\rho}}
Let $ u\#v\in \sigma_{X,r,s}(E_{X,r,s}) $. Then let $ u=u_0ru_1ru_2r\cdots ru_k $, for $ u_1r\cdots u_{k-1}r\in \lrepeat_{X,r,s}^* $, and $ v=v_0sv_1sv_2s\cdots sv_l $, for $ v_1sv_2s\cdots v_{l-1}s\in \rrepeat_{X,r,s}^* $. That is, $ u_0\lhash u_k\# v_0\rhash v_k\in E_{X,r,s} $. This implies that there exists $ \lw $ (and $ \rw $), such that $ u_0 $ ($ v_0 $) and $ u_k $ ($ v_l $) are respectively the first and the last sub-$ r $ (sub-$ s $) blocks of $ \lw $ ($ \rw $). Let the production rule sequence $ X\xRightarrow{*}\lw X \rw $ be denoted by $ \rho_0 $. 

Then by the definition of $ \lrepeat_{X,r,s} $, we have production rule sequences $ \lrho_i\coloneqq X\xRightarrow{*} e_iXf_i $, such that $ u_i $ is a sub-$ r $ block of $ e_i $, and $ f_i\in\Max{s}  $. Similarly, there are production rule sequences $ \rrho_j\coloneqq X\xRightarrow{*} g_jXh_i $ such that $ v_j $ is a sub-$ s $ block of $ h_j $, and $ g_j\in\Max{r}  $.

Then the derivation sequence $ \rho_0\lrho_1\cdots \lrho_k\rrho_{1}\cdots \rrho_l\rho_0 $ gives 
$ X\derivs u'Xv' $, where $ u'=\lw  e_1e_2\cdots e_k g_1\cdots g_l   \lw $ and $ v'=\rw  f_1f_2\cdots f_k h_1\cdots h_l   \rw $. It is easy to see that $ u\bo u' $ and $ v\bo v' $.

\end{proof}

\subsection{Proof of \cref{cfl-correctness}}
\cflCorrectness*
\begin{proof}
Of course, for zero expansions, there is nothing to prove, so suppose we have a
derivation in $\cG'$ with $k$ expansions and let $t$ be a derivation tree
obtained by an expansion step from $t_0$ by replacing the $X$-labeled node $x$.
Moreover, let $X\derivs uwv$ be the derivation using new productions inserted
at $x$.  Then by construction, we know that $u\#v\in
\sigma_X(E_X)\bd$.  By \cref{embedding-fresh-letter}, this implies
that there is a word $u'\#v'\in\sigma_X(E_X)$ with $u',v'\in\Sigma_p^*$, $u\bo
u'$, and $v\bo v'$ Therefore, by \cref{grammar-correctness}, there exist
$u'',v''\in\Sigma_p^*$ with $u'\bo u''$ and $v'\bo v''$ and a derivation
$X\derivs u''Xv''$ in $\cG$. In particular, there is a derivation $X\derivs
u''wv''$ in  $\cG$. Now consider the derivation $t_1$ obtained from $t_0$ by
replacing $x$ with the derivation $X\derivs u''wv''$.  Then $t_1$ derives a new
word of the form $\alpha u''\beta v''\gamma$, where $\alpha\beta\gamma$ is the
word derived by $t_0$. Since $t_1$ only uses productions from $\cG$ in addition
to those in $t_0$, we know that $t_1$ needs just $k-1$ expansions. Hence, we
know by induction that $\alpha u''\beta v''\gamma\in\lang(\cG)\bd$.
Since $\bo$ is multiplicative, we have $\alpha u\beta v\gamma\bo\alpha u'\beta
v'\gamma\bo \alpha u''\beta v''\gamma\in\lang(\cG)\bd$.  Since
$\alpha u\beta v\gamma$ is the word generated by $t$ and belongs to
$\lang(\cG)\bd$, this completes the proof.
\end{proof}

\subsection{Proof of \cref{cfl-acyclicity}}
\cflAcyclicity*
\begin{proof}
Consider
a derivation tree $t$ in $\cG'$. We pick $t$ so that it minimizes the number of
nodes whose label repeats below them.  Note that each path in $t$ alternates
between productions from $\cG$ and segments using newly introduced productions.
By induction, we may assume that within each segment, no non-terminal repeats.
Now observe that if any new non-terminal that occurs in two segments, then these
two segments must come from the same grammar $\egrammar'_X$ and thus $X$ must
repeat on that path. Therefore, if there is any repetition, there is also a
repetition of some non-terminal $X$ of $\cG$. On the path where $X$ repeats,
pick the top-most and the lowest occurrence of $X$.  Between these two, we have
a derivation $X\derivs_{\cG'} uXv$.  This derivation can be replaced by a
single derivation using new productions. Moreover, by our choice of occurrences
of $X$, this replacement will not introduce cross-segment repetitions. Thus, we
obtain a new derivation with (at least) one fewer repetition. By applying this
argument again and again, we arrive at a derivation with no repetitions. We have thus shown:
For every derivation tree $ t $ in $ \cG' $ with repetition of non-terminals, 
there exists an equivalent acyclic tree in $ \cG' $. This, with \cref{cfl-correctness}, 
gives the desired result. 
\end{proof}

\end{document}